\documentclass[acmsmall, screen, nonacm]{acmart}

\usepackage{cases}
\usepackage{IEEEtrantools}

\newcommand{\bs}{\boldsymbol}
\newcommand{\mb}{\mathbb}
\newcommand{\mc}{\mathcal}

\newcommand{\blank}{\;\;}
\renewcommand{\l}{\langle}
\renewcommand{\r}{\rangle}
\newcommand{\ua}{\rangle}
\DeclareMathOperator*{\argmax}{arg\,max}

\begin{document}

\title{Exploiting Scheduling Flexibility via State-Based Scheduling When Guaranteeing Worst-Case Services}

\author{Yike Xu}
\authornote{Yike Xu is the corresponding author.}
\email{yike\_xu@hotmail.com}

\author{Mark S. Andersland}
\email{mark-andersland@uiowa.edu}

\begin{abstract}
Even when providing long-run, worst-case guarantees to competing flows of unit-sized tasks, a slot-timed, constant-capacity server's scheduler may retain significant, short-run, scheduling flexibility. Existing worst-case scheduling frameworks offer only limited opportunities to characterize and exploit this flexibility. We introduce a state-based framework that overcomes these limitations. Each flow's guarantee is modeled as a worst-case service that can be updated as tasks arrive and are served. Taking all flows' worst-case services as a collective state, a state-based scheduler ensures, from slot to slot, transitions between schedulable states. This constrains its scheduling flexibility to a polytope consisting of all feasible schedules that preserve schedulability.

We fully characterize this polytope, enabling scheduling flexibility to be fully exploited. But, as our framework is general, full exploitation is computationally complex. To reduce complexity, we show: that when feasible schedules exist, at least one can be efficiently identified by simply maximizing the server's capacity slack; that a special class of worst-case services, min-plus services, can be efficiently specified and updated using the min-plus algebra; and that efficiency can be further improved by restricting attention to a min-plus service subclass, dual-curve services. This last specialization turns out to be a dynamic extension of service curves that approaches near practical viability while maintaining all features essential to our framework.
\end{abstract}



\maketitle

\section{Introduction}\label{S:introduction}

Given a slot-timed, constant-capacity server, what short-run scheduling decisions must be made to provide long-run, worst-case guarantees to competing flows of unit-sized tasks? Although this problem has been extensively investigated, questions remain. Consider the following example:

\begin{example}\label{Ex:appetizer}
	Suppose that two flows compete for service from a $c$-task-per-slot server. In slot~$t$, each has $50c$ tasks queued in its buffer, but the tasks from flow~$1$ must be served before $t+99$, while those from flow~$2$, before $t+100$. Clearly, during interval $[t, t+100)$, the server needs to be work-conserving in the sense that it always serves as many tasks as it can. If, in each slot, the server's capacity is split half-and-half between the two flows, $c/2$ tasks from flow~$1$ will miss their deadlines. If instead, tasks are scheduled earliest-deadline-first (EDF), all deadlines will be met as flow~$1$ will be served exclusively during interval $[t, t+50)$ and flow~$2$, during $[t+50, t+100)$. In fact, if the scheduler is smart enough, it can freely apportion service to both flows during interval $[t, t+99)$ as long as it leaves $c$~tasks from flow~$2$ to be served in slot $t+99$.	
\end{example}

Of course, this example is highly simplified. In fact, outlining the smart scheduler's flexibility cannot be so easy if, for instance, task deadlines are not so uniform, or some tasks are allowed to arrive after slot~$t$ but before their deadlines. Nonetheless, the example suffices to highlight that, even when providing long-run, worst-case guarantees to competing flows, the scheduler may retain a significant degree of short-run flexibility. In this paper, we introduce a state-based framework in which this flexibility can be characterized and exploited. Our principal contributions are as follows.

\begin{itemize}
	\item {\bf Modeling Guarantees as States:} A natural model of a flow's long-run, worst-case guarantee is a worst-case service that maps each realization of its arrival process to a worst-case service process. We show that this map can be updated by discounting the service already received. So worst-case services are states that can be updated as tasks arrive and are served.
	
	\item {\bf Introducing State-Based Scheduling:} Taking all flows' worst-case services as a collective state, we find the condition that this collective state satisfies if and only if it is schedulable. A state-based scheduler must ensure, from slot to slot, transitions between schedulable states.	This constrains the scheduler's flexibility by constraining its state evolution.
	
	\item {\bf Characterizing Scheduling Flexibility:} In each slot, the state-based scheduler's flexibility is constrained to a feasible polytope, consisting of all feasible schedules that preserve schedulability. We assemble this polytope from a sequence of permutohedral slices.
	
	\item {\bf Exploiting Scheduling Flexibility:} Schedules can be freely selected from the feasible polytope to enforce, for instance, priority or fairness criteria. We show that this flexibility can also be traded for efficiency by selecting max-slack schedules or their per-class extensions.
	
	\item {\bf Reducing Complexity:} Our framework is computationally complex. We identify three complexity-reducing specializations that culminate in dual-curve services, a dynamic extension of service curves that approaches near practical viability while maintaining all features essential to our framework. To illustrate this, we revisit Example~\ref{Ex:appetizer}  at the end of the paper.
\end{itemize}

\subsection{A State-Based Framework}\label{SS:framework}

In \cite{Cruz:1991A, Cruz:1991B, Parekh:1993, Parekh:1994}, cumulative curves were introduced to characterize arrival and service processes. They become cumulative vectors in slot-timed systems. In particular, arrival vectors can be used to count a flow's cumulative task arrivals, and departure vectors, its cumulative served tasks. We extend the definition of arrival vectors to include initially queued tasks and define a worst-case service to be a map from each such queued arrival vector to a worst-case departure vector. Including queued tasks in the arrival counts enables worst-case services to be defined not only when buffers are empty but also when they are non-empty. This is necessary for these services to be states that can be updated as tasks arrive and are served.

Taking all flows' worst-case services as a collective state, the key to state-based scheduling is finding the schedulability condition on this collective state necessary and sufficient to ensure that all flows' worst-case services can be guaranteed. Once this condition is found, in each slot, it can be used, on the one hand, to admit or deny new service requests, and on the other, to identify all feasible schedules that preserve schedulability. To find the schedulability condition, we introduce the concept of the spectrum of a worst-case service. It specifies, during each slot interval, the least capacity that must be reserved to guarantee the worst-case service. We use this concept to formulate the schedulability condition and then use the condition to identify all feasible schedules. The set of feasible schedules turns out to be a feasible polytope. A principal constraint on this polytope is determined by a baseline function that specifies the least number of tasks that must be served from any given subset of flows. As the baseline function is supermodular, we show that, when the total service is fixed, the resulting slice of the feasible polytope is a permutohedron, a special polytope from polymatroid theory \cite{Edmonds:1970}.

\subsection{Three Specializations}\label{SS:specailizations}

A downside of our framework's generality is its complexity. This complexity is two-fold. On the one hand, to fully exploit the flexibility of selecting any feasible schedule, the scheduler must fully determine the feasible polytope, which is combinatorially difficult. On the other hand, worst-case services, as uncountably infinite, full-blown maps between cumulative vectors, are challenging to specify and update. To address these difficulties, but also for their own merits, we consider three specializations: max-slack schedules, min-plus services, and dual-curve services.

Max-slack schedules maximize the server's capacity slack, that is, they leave maximum room for the server to admit new service requests. When the total service is fixed, the max-slack schedule is feasible if a feasible schedule exists. Moreover, aggregating flows into classes such that, intra-class, flows are max-slack scheduled, enables intermediate tradeoffs between flexibility and efficiency, because feasible inter-class schedules still form permutohedra, but of lower dimension.

Min-plus services are motivated by the insight that, among all worst-case services that share the same spectrum, there exists a maximum and this maximum can be constructed using the min-plus algebra. So, by construction, min-plus services can be completely identified by their spectra. Accordingly, to specify and update these services, instead of an uncountably infinite map between cumulative vectors, we need only specify and update a countably infinite spectral matrix.

An alternative definition of min-plus services enables a simplification of their update rule. Motivated by this simplification, their efficiency can be further improved by restricting attention to a min-plus service subclass, dual-curve services. To specify and update these services, instead of a spectral matrix, we need only specify and update a pair of cumulative vectors, one static and one dynamic. We call them dual-curve, as opposed to dual-vector, services to highlight their connection to service curves. Each service curve, according to \cite{Cruz:1995}, can be specified by a static cumulative vector. Adding a dynamic vector, yields the dynamic extension, a dual-curve service. It is well known that service curves can be guaranteed by EDF scheduling \cite{Sariowan:1999}. When applied to dual-curve services, EDF scheduling results in max-slack schedules. In contrast, state-based scheduling is able to identify all feasible schedules, max-slack or non-max-slack.

The rest of the paper is organized as follows. In Section~\ref{S:serverModel}, we introduce our service model. In Section~\ref{S:WCS}, we define worst-case services and their spectra. In Sections~\ref{S:schedulability} and \ref{S:feasible}, we introduce state-based scheduling, find the schedulability condition, and then use this condition to identify all feasible schedules. In Sections~\ref{S:MSS}, \ref{S:MPS}, and \ref{S:DVS}, we introduce, respectively, max-slack schedules, min-plus services, and dual-curve services. To save space, all proofs are relegated to appendices.\footnote{In particular, all proofs of the theorems appearing in Sections \ref{S:WCS}-\ref{S:MPS} can be found in Appendices \ref{A:WCSProof}-\ref{A:MPSProof}, respectively.} Also, some alternative proofs, as well as more results, comments, and examples, can be found in an extended exposition \cite{Xu:2021}.

\section{The Service Model}\label{S:serverModel}

Our service model is discrete in that time is slotted and all tasks are of unit size. As illustrated in Fig.~\ref{F:serverModel}, for each flow, indexed by $\omega\in\Omega := \{1,2,\ldots,n\}$, at the beginning of slot~$t$, $a^\omega$ tasks arrive and are immediately queued behind the $b^\omega$ tasks left unserved prior to slot~$t$. During slot~$t$, the scheduler determines $d^\omega$, the number of tasks to be served. As tasks cannot be served before they arrive,
\begin{equation}\label{E:causalityA}
	d^\omega \leq q^\omega := a^\omega+b^\omega,
\end{equation}
where $q^\omega$ is the number of tasks queued in flow~$\omega$'s buffer. Within each flow, the service order is first come, first served, so that, at the end of slot~$t$, the first $d^\omega$ tasks queued in flow~$\omega$'s buffer depart and are served, leaving
\begin{equation}\label{E:bUpdate}
	\dot{b}^\omega = q^\omega-d^\omega
\end{equation}
tasks unserved prior to slot $t+1$. As in (\ref{E:causalityA}) and (\ref{E:bUpdate}), unless noted otherwise, all variables are implicitly indexed by current slot~$t$. To index next slot $t+1$, we add a dot, as in $\dot{b}^\omega$.

\begin{figure}[t]
	\centering \scalebox{0.972}{\includegraphics{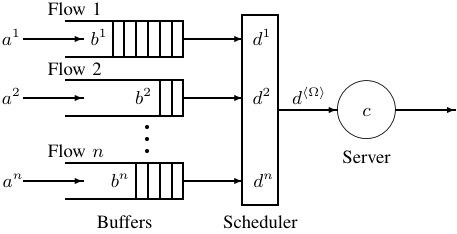}} \caption{In our service model, the tasks arrive from $n$ distinct flows. Each flow's tasks are buffered separately, either physically or virtually. All flows, however, share a single $c$-task-per-slot server, access to which is controlled by a scheduler.}
	\label{F:serverModel}
	\Description{The service model, composed of a buffer system, a scheduler and a constant-capacity server.}
\end{figure}

The scheduler's choices are constrained. Denoting the ensemble of flow variables, $[x^1,x^2,\ldots,x^n]$, by $x^{[\Omega]}$, and the sum, $\sum_{\omega \in \Omega} x^\omega$, by $x^{\l\Omega\r}$, it follows that the selected schedule, $d^{[\Omega]}$, must satisfy a {\it \textbf{causality constraint}},
\begin{equation}\label{E:causalityB}
	d^{[\Omega]} \leq q^{[\Omega]} = a^{[\Omega]}+b^{[\Omega]},
\end{equation}
which restates (\ref{E:causalityA}) in ensemble form, and a {\it \textbf{capacity constraint}},
\begin{equation}\label{E:capacity}
	d^{\l\Omega\r} \leq c,
\end{equation}
which requires that the number of tasks served not exceed the server's capacity. We call a $d^{[\Omega]}$ that satisfies both (\ref{E:causalityB}) and (\ref{E:capacity}) a {\it \textbf{valid schedule}}. The scheduler may only select among valid schedules.

\section{Worst-Case Services}\label{S:WCS}

In this section, we define worst-case services. We show that these services are states that can be updated as tasks arrive and are served. We also introduce the spectrum of a worst-case service as a characterization of the capacities that must be reserved to guarantee it. As a worst-case service is updatable, so is its spectrum.

\subsection{Definition}\label{SS:WCSDef}

A {\it \textbf{cumulative vector}} is a semi-infinite, non-decreasing vector that starts with $0$, with elements in $\mb{N}^+  := \mb{N} \cup \{\infty\}$, where $\mb{N}$ denotes the set of natural numbers. Let $\mb{U}$ be the set of all cumulative vectors. Then $\bs{x}=[x_j]_{j \in \mb{N}}=[x_0,x_1,x_2,\ldots] \in \mb{U}$ if $x_0=0$, and for all $j \in \mb{N}$, $x_j \in \mb{N}^+$ and $x_j \leq x_{j+1}$. Important, useful subsets of $\mb{U}$ include
\begin{equation}\label{E:UxEntries}
	\text{(a)}\blank\mb{U} \ua x := \{\bs{x}\in\mb{U} | x_1 \geq x\}, \blank\text{and}\blank
	\text{(b)}\blank\mb{U}|x := \{\bs{x}\in\mb{U} | x_1=x\}.
\end{equation}

For each flow, we use an {\it \textbf{arrival vector}}, $\bs{a}=[a_j]_{j\in\mb{N}} \in \mb{U}$, to count its cumulative task arrivals, and a {\it \textbf{departure vector}}, $\bs{d}=[d_j]_{j\in\mb{N}} \in \mb{U}$, to count its cumulative served tasks. In particular, for all $j>0$, $a_j$ and $d_j$ count the tasks that, respectively, arrive and are served during interval $[t, t+j)$, that is, from slot~$t$ to $t+j-1$. Here, as we are referencing a generic flow, we suppress its index. Notice that $\bs{a}$ and $\bs{d}$ completely characterize the arrival process of, and the service process to, the flow. To be precise, let
\begin{equation}\label{E:tau}
	\tau_h(\bs{x}):= \max \{j\in\mb{N}^+|x_j < h\}, \blank\blank h=1,2,3,\ldots.
\end{equation}
Then the $h$th task in $\bs{a}$ arrives in slot $t+\tau_h(\bs{a})$, while the $h$th task in $\bs{d}$ is served in slot $t+\tau_h(\bs{d})$. According to (\ref{E:tau}), $\tau_h(\bs{x}) = \infty$ if $h > x_\infty := \lim_{j \to \infty} x_j$. For instance, if $h > d_\infty$, $\tau_h(\bs{d}) = \infty$, implying that the $h$th task, even if it exists, is never served.

During each $[t, t+j)$, as the number of tasks served cannot exceed the number of arrivals plus the number of tasks left unserved prior to slot $t$,
\begin{subnumcases}{\label{E:arrivalCases}
	d_j \leq q_j :=}
		0       & if $j=0$,\label{E:arrivalCaseA}\\
		a_j + b & if $j>0$,\label{E:arrivalCaseB}
\end{subnumcases}
or in vector form
\begin{equation}\label{E:vecArrival}
	\bs{d} \leq \bs{q} := \bs{a}+b\bs\delta,
\end{equation}
where $\bs\delta=[\delta_j]_{j \in \mb{N}} := [0, 1, 1, \ldots] \in \mb{U}$, that is, $\delta_0 := 0$ and $\delta_j := 1$ for all $j>0$. Clearly, (\ref{E:vecArrival}) is the vector extension of (\ref{E:causalityA}). By definition, $\bs{q} \in \mb{U} \ua b$, where $\mb{U} \ua b$ is defined according to (\ref{E:UxEntries})-(a). As $b$ is entirely fixed by the flow's past, $\bs{q}$ can be viewed as a bijective function of $\bs{a}$ mapping $\mb{U}$ to $\mb{U} \ua b$. Compared to $\bs{a}$, it is as if the $b$ tasks already queued in the buffer are miscounted by $\bs{q}$ as new arrivals. So we call $\bs{q}$ the {\it \textbf{queued arrival vector}}. Since the server maps each $\bs{q}$ to some $\bs{d}$, a natural way to specify a service is in terms of a map from each $\bs{q}$ to a worst-case $\bs{d}$.

\begin{definition}\label{D:WCS}
	For a flow with $b$ tasks left unserved prior to slot~$t$, $\bs\psi: \mb{U} \ua b \rightarrow \mb{U}$ is a \textbf{worst-case service} if
	\begin{equation}\label{E:WCSDef}
		\bs\psi(\bs{q}) \leq \bs{q} \blank\blank \forall \bs{q} \in \mb{U} \ua b.
	\end{equation}
	The flow is said to be guaranteed worst-case service $\bs\psi$ if
	\begin{equation}\label{E:WCSDef+}
		\bs{d} \geq \bs\psi(\bs{q}) \blank\blank \forall \bs{q} \in \mb{U} \ua b.
	\end{equation}
\end{definition}

Since $\bs\psi$ is conditioned on $b$, whenever we refer to $\bs\psi$, we implicitly refer to the pair, $(\bs\psi, b)$. As illustrated in Fig.~\ref{F:Bounds}, to guarantee $\bs\psi$, $\bs{d}$ must lie between $\bs{q}$ and $\bs\psi(\bs{q})$, which explains why we need (\ref{E:WCSDef}) in Definition \ref{D:WCS}. It can also be seen in Fig.~\ref{F:Bounds} that flow backlogs and task delays are bounded by, respectively, the vertical and horizontal distances between $\bs{q}$ and $\bs\psi(\bs{q})$. We can, in fact, even design worst-case services to guarantee given backlog or delay bounds.\footnote{For more on these performance bounds and their use in analyzing and designing worst-case services, see Appendix \ref{A:performance}.} Notice that, by definition, to specify a worst-case service, we need only identify a $\bs\psi(\bs{q}) \leq \bs{q}$ for each $\bs{q}\in\mb{U} \ua b$. This is not practical in general because $\mb{U} \ua b$ is uncountably infinite. Nonetheless, the theoretical possibility of specifying services so broadly itself underlies our framework's generality.

\begin{figure}[t]
	\centering \scalebox{0.875}{\includegraphics{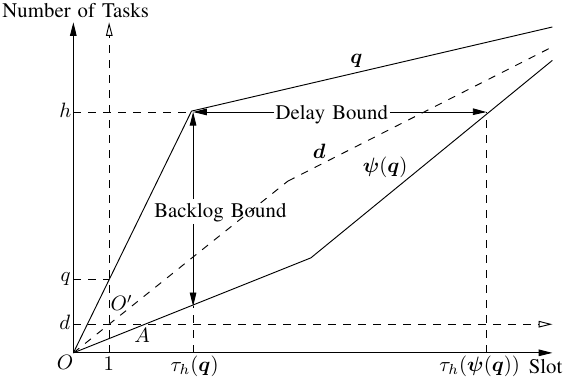}} \caption{When worst-case service $\bs\psi$ is guaranteed, $\bs{d}$ must lie between $\bs{q}$ and $\bs\psi(\bs{q})$. Flow backlogs and task delays are bounded by, respectively, the vertical and horizontal distances between $\bs{q}$ and $\bs\psi(\bs{q})$. Updating $\bs\psi$ reduces to re-expressing $\bs{\psi}(\bs{q})$ in the translated coordinate frame with origin $O'$.}
	\label{F:Bounds}
	\Description{Performance bounds and the update rule for a worst-case service.}
\end{figure}

\subsection{The Update Rule}\label{SS:WCSUpdate}

When $a$ tasks arrive in slot~$t$, the set of possible $\bs{a}$'s shrinks from $\mb{U}$ to $\mb{U}|a$, so due to (\ref{E:vecArrival}) and (\ref{E:causalityA}), that of possible $\bs{q}$'s shrinks from $\mb{U} \ua b$ to $\mb{U}|q$, where $\mb{U}|a$ and $\mb{U}|q$ are defined according to (\ref{E:UxEntries})-(b). In this paper, as in practice, we assume the scheduler to be {\it \textbf{causal}} in the sense that it cannot foresee future arrivals. So it cannot foresee which $\bs{q} \in \mb{U}|q$ will be realized, and thus must ensure that
\begin{equation}\label{E:WCSImmediate}
	d \geq p := \max_{\bs{q} \in \mb{U}|q} \psi_1(\bs{q}),
\end{equation}
to ensure that $\bs{d} \geq \bs\psi(\bs{q})$ can be guaranteed no matter which $\bs{q}\in\mb{U}|q$ is realized. According to (\ref{E:WCSDef}),
\begin{equation}\label{E:p<q}
	p \leq \max_{\bs{q} \in \mb{U}|q} q_1 = q.
\end{equation}

The immediate portion of $\bs\psi$ to be met by $d$ is denoted by $p$ in (\ref{E:WCSImmediate}). But what about the remaining portion? As we will see, it turns out to be yet another worst-case service. Intuitively, as illustrated in Fig.~\ref{F:Bounds}, after $d$ tasks have been served, $\bs\psi(\bs{q})$ can be re-expressed in a translated coordinate frame in which the origin moves from $O$ to $O'$ at $(1,d)$ in the original frame. Discounting the immediate portion met by $d$, $\bs{\psi}(\bs{q})$ is truncated in this new frame. A new worst-case departure vector can then, roughly speaking, be constructed by splicing the line segment $\overline{O'A}$ to the truncated $\bs{\psi}(\bs{q})$, that is, by replacing $\overline{OA}$ with $\overline{O'A}$.

To formalize this intuition, observe first that since the counting process for $\bs{d}$ starts from slot~$t$ while that for $\bs{\dot{d}}$ starts from $t+1$, $d_{j+1}$ and $\dot{d}_j$ count, respectively, the number of tasks that are served during intervals $[t, t+j+1)$ and $[t+1, t+j+1)$. Therefore, when $d$ tasks are served in slot~$t$, for all $j\in\mb{N}$,
$d_{j+1}=\dot{d}_j+d$. This can be rewritten in vector form as
\begin{equation}\label{E:vecdUpdate}
	\bs{d} = \mc{R}\bs{\dot{d}}+d\bs\delta,
\end{equation}
where $\mc{R}: \mb{U} \rightarrow \mb{U}|0$ is the right-shift operator defined by $[\mc{R}\bs{x}]_{j+1} := x_j$ for all $j\in\mb{N}$. Similarly, when $a$ tasks arrive in slot~$t$, $a_{j+1}=\dot{a}_j+a$, so using (\ref{E:arrivalCaseB}), (\ref{E:causalityA}), and (\ref{E:bUpdate}), we have
\begin{subnumcases}{\label{E:qUpdateCases} 
	q_{j+1} = a_{j+1}+b = \dot{a}_j+a+b = \dot{a}_j+q =}
		q           & if $j=0$,\label{E:qUpdateCaseA}\\
		\dot{q}_j+d & if $j>0$,\label{E:qUpdateCaseB}
\end{subnumcases}
where the last equality holds because $\dot{a}_0 = 0$, and if $j>0$, as $\dot{q}_j = \dot{a}_j+\dot{b}$ and $\dot{b} = q-d$, $\dot{a}_j+q = \dot{q}_j-\dot{b}+q = \dot{q}_j+d$. Since $\bs{\dot{q}} = \bs{\dot{a}}+\dot{b}\bs\delta$, (\ref{E:qUpdateCases}) can be rewritten in vector form as
\begin{equation}\label{E:vecqUpdate}
	\bs{q} = \mc{R}\bs{\dot{a}}+q\bs\delta = \mc{R}(\bs{\dot{q}}-\dot{b}\bs\delta)+q\bs\delta.
\end{equation}

Observe next that if $\bs\psi$ is guaranteed, according to Definition \ref{D:WCS}, for all $\bs{q} \in \mb{U}|q$, $\bs{d} \geq \bs\psi(\bs{q})$. So, according to (\ref{E:vecdUpdate}), $\mc{R}\bs{\dot{d}} \geq \bs\psi(\bs{q})-d\bs\delta$. Let $x^+$ denote $\max \{x, 0\}$, and $\mc{R}^{-1}: \mb{U}|0 \rightarrow \mb{U}$ denote the inverse of $\mc{R}$, that is, $[\mc{R}^{-1}\bs{x}]_j= x_{j+1}$ for all $j\in\mb{N}$. Then, as $\mc{R}\bs{\dot{d}} \in \mb{U}|0$, $\mc{R}\bs{\dot{d}} \geq \bs\psi(\bs{q})-d\bs\delta$ implies that $\bs{\dot{d}} \geq \mc{R}^{-1}(\bs\psi(\bs{q})-d\bs\delta)^+$. Notice that (\ref{E:vecqUpdate}) establishes that $\bs{q}$ is a bijective function of $\bs{\dot{q}}$, mapping $\mb{U}\ua\dot{b}$ to $\mb{U}|q$. So $\mc{R}^{-1}(\bs\psi(\bs{q})-d\bs\delta)^+$ is also a function of $\bs{\dot{q}}$. This fact leads to the following update rule.

\begin{theorem}\label{T:WCSUpdate}
	For a flow with $b$ tasks left unserved prior to slot~$t$, when $a$ tasks arrive and $d$ tasks are served in slot~$t$, if (\ref{E:causalityA}) and (\ref{E:WCSImmediate}) hold, that is, if $q \geq d \geq p$, the flow is guaranteed worst-case service $\bs\psi$ if and only if
	\begin{equation} \label{E:WCSUpdateB}
		\forall \bs{\dot{q}} \in \mb{U}\ua\dot{b},\blank\blank
		\bs{\dot{d}} \geq \bs{\dot\psi}(\bs{\dot{q}}) := \mc{R}^{-1}(\bs\psi(\bs{q})-d\bs\delta)^+ = \mc{R}^{-1}(\bs\psi(\mc{R}(\bs{\dot{q}}-\dot{b}\bs\delta)+q\bs\delta)-d\bs\delta)^+,
	\end{equation}
	where $\bs{\dot\psi}$ is a worst-case service for the flow in slot $t+1$.
\end{theorem}

According to this theorem, $\bs\psi$ is a state variable that can be updated to $\bs{\dot\psi}$, the remaining portion of $\bs\psi$ to be guaranteed after slot~$t$. Recall that $\bs\psi$ is conditioned on $b$, so $\bs{\dot\psi}$ is conditioned on $\dot{b}$. According to (\ref{E:bUpdate}), $b$ is also a state variable that can be updated to $\dot{b}$. Therefore, whenever we update $\bs\psi$ to $\bs{\dot\psi}$, we implicitly update $(\bs\psi, b)$ to $(\bs{\dot\psi}, \dot{b})$ through (\ref{E:WCSUpdateB}) and (\ref{E:bUpdate}).

\subsection{The Spectrum}\label{SS:signatures}

Full-blown maps between cumulative vectors are difficult to visualize. The concept of spectrum helps us distill what is essential for our purposes. The key question is: for the causal scheduler to guarantee worst-case service $\bs\psi$, what capacity must be reserved during each slot interval? To guarantee $\bs\psi$, according to (\ref{E:WCSDef+}), $d_j \geq \psi_j(\bs{q})$ tasks must be served during interval $[t,t+j)$. But according to (\ref{E:arrivalCases}), only $d_i \leq q_i$ tasks can be served during $[t, t+i)$. It follows that $(d_j-d_i)^+ \geq (\psi_j(\bs{q})-q_i)^+$, where $(d_j-d_i)^+$ counts the number of tasks that are served during $[t+i, t+j)$, and is $0$ by default if $i \geq j$. This observation, which holds for all $\bs{q} \in \mb{U} \ua b$, motivates the following definition.

\begin{definition}\label{D:spectrum}
	Given worst-case service $\bs\psi$, for all $i, j \in \mb{N}$, the \textbf{spectral value} of $\bs\psi$ over interval $[t+i, t+j)$ is
	\begin{equation}\label{E:signature}
		\lambda_{ij}(\bs\psi):= \max_{\bs{q} \in \mb{U} \ua b} (\psi_j(\bs{q})-q_i)^+.
	\end{equation}
	We call the collection of all such values the \textbf{spectrum} of $\bs\psi$.
\end{definition}

By construction, $\lambda_{ij}(\bs\psi)$ is the least capacity that must be reserved during interval $[t+i, t+j)$ to ensure that $\bs{d} \geq \bs\psi(\bs{q})$ can be guaranteed no matter which $\bs{q}\in\mb{U}\ua b$ is realized. It is immediate from (\ref{E:signature}) that $\lambda_{ij}(\bs\psi) \geq \lambda_{ij}(\bs{\psi'})$ if $\bs\psi \geq \bs{\psi'}$, that is, if $\bs\psi(\bs{q}) \geq \bs{\psi'}(\bs{q})$ for all $\bs{q}\in\mb{U}\ua b$. When no confusion can be introduced, we will denote $\lambda_{ij}(\bs\psi)$ by $\lambda_{ij}$. The next theorem lists some basic properties of $\lambda_{ij}$.

\begin{theorem}\label{T:SigProp}
	For all $i, j \in \mb{N}$,
	\begin{equation}\label{E:SigPropEntries}
		\textup{(a)}\blank \lambda_{ij} = 0 \textup{ if } i \geq j,
		\blank
		\textup{(b)}\blank \lambda_{ij} \leq \lambda_{i, j+1},
		\blank
		\textup{(c)}\blank \lambda_{ij} \geq \lambda_{i+1, j},
		\blank\textup{and}\blank
		\textup{(d)}\blank \lambda_{ij} \leq (\lambda_{0j}-b)^+ \textup{ if } i>0.
	\end{equation}
\end{theorem}

\subsubsection*{\textup{[}Updating the Spectrum\textup{]}} According to Theorem~\ref{T:WCSUpdate}, given $q \geq d \geq p$, we can update $\bs\psi$ to $\bs{\dot\psi}$. Denote $\lambda_{ij}(\bs{\dot\psi})$ by $\dot\lambda_{ij}$. Then, for all $i,j\in\mb{N}$, using (\ref{E:signature}) and (\ref{E:WCSUpdateB}), we have
\[
	\dot\lambda_{ij} = \lambda_{ij}(\bs{\dot\psi}) = \max_{\bs{\dot{q}} \in \mb{U}\ua \dot{b}} (\dot\psi_j(\bs{\dot{q}})-\dot{q}_i)^+ = \max_{\bs{q} \in \mb{U}|q} ([\mc{R}^{-1}(\bs\psi(\bs{q})-d\bs\delta)^+]_j-\dot{q}_i)^+ = \max_{\bs{q} \in \mb{U}|q} (\psi_{j+1}(\bs{q})-d-\dot{q}_i)^+,
\]
where the second equality holds because, due to (\ref{E:vecqUpdate}), $\bs{\dot{q}} \in \mb{U}\ua \dot{b}$ is equivalent to $\bs{q} \in \mb{U}|q$. This implies that if $i=0$, as $\dot{q}_0 = 0$,
\begin{equation}\label{E:SigUpdProof+}
	\dot\lambda_{0j} = \max_{\bs{q} \in \mb{U}|q} (\psi_{j+1}(\bs{q})-d)^+ = \left(\max_{\bs{q} \in \mb{U}|q} \psi_{j+1}(\bs{q})-d\right)^+,
\end{equation}
and if $i>0$, due to (\ref{E:qUpdateCaseB}),
\begin{equation}\label{E:SigUpdProof++}
	\dot\lambda_{ij} = \max_{\bs{q} \in \mb{U}|q} (\psi_{j+1}(\bs{q})-q_{i+1})^+.
\end{equation}
Let
\begin{equation}\label{E:CondSig}
	\lambda_{ij}(\bs\psi|q) := \max_{\bs{q} \in \mb{U}|q} (\psi_j(\bs{q})-q_i)^+,
\end{equation}
and denote it by $\hat\lambda_{ij}$. Rewriting (\ref{E:SigUpdProof+}) and (\ref{E:SigUpdProof++}) in terms of $\hat\lambda_{ij}$, we obtain the next theorem.\footnote{For an intuitive interpretation of Theorem \ref{T:SigUpdate}, see Appendix \ref{A:SigUpdateInterpret}.}

\begin{theorem}\label{T:SigUpdate}
	Given $\bs{\dot\psi}$ in Theorem~\ref{T:WCSUpdate}, for all $i, j \in \mb{N}$,	\begin{subnumcases}{\label{E:SigUpdateCases}
		\dot\lambda_{ij}=}
			(\hat\lambda_{0, j+1}-d)^+ & \textup{if $i=0$,}\label{E:SigUpdateCaseA}\\
			\hat\lambda_{i+1, j+1}     & \textup{if $i>0$.}\label{E:SigUpdateCaseB}
	\end{subnumcases}
\end{theorem}

We call $\hat\lambda_{ij}$ the {\it \textbf{conditional spectral value}} of $\bs\psi$ over interval $[t+i, t+j)$, and the collection of all such values, the {\it \textbf{conditional spectrum}} of $\bs\psi$. Comparing (\ref{E:CondSig}) to (\ref{E:signature}), it is seen that the only difference is that $\bs{q}$'s range shrinks from $\mb{U} \ua b$ to $\mb{U} | q$, so for all $i, j \in \mb{N}$,
\begin{equation}\label{E:CondSigPropF}
	\hat\lambda_{ij} \leq \lambda_{ij}.
\end{equation}
According to (\ref{E:CondSig}) and (\ref{E:WCSImmediate}),
\begin{equation}\label{E:p-lambda}
	\hat\lambda_{01} = \max_{\bs{q} \in \mb{U}|q} \psi_1(\bs{q}) = p.
\end{equation}
Also, according to (\ref{E:CondSig}),
\begin{equation}\label{E:CondSigPropE}
	\hat\lambda_{1j} = \max_{\bs{q} \in \mb{U}|q} (\psi_j(\bs{q})-q_1)^+ = \left(\max_{\bs{q} \in \mb{U}|q} \psi_j(\bs{q})-q\right)^+ = (\hat\lambda_{0j}-q)^+,
\end{equation}
Comparing this to (\ref{E:SigUpdateCaseA}), it is immediate that, given any $d \leq q$,
\begin{equation}\label{E:lambdaBound}
	\dot\lambda_{0j} \geq \hat\lambda_{1, j+1}.
\end{equation}
the lower bound of which is achieved when $d=q$. In addition, paralleling Theorem~\ref{T:SigProp}, the next theorem lists more basic properties of $\hat\lambda_{ij}$.

\begin{theorem}\label{T:CondSigProp}
	For all $i, j \in \mb{N}$,
	\begin{equation}\label{E:CondSigPropEntries}
		\textup{(a)}\blank \hat\lambda_{ij} = 0 \textup{ if } i \geq j,
		\blank
		\textup{(b)}\blank \hat\lambda_{ij} \leq \hat\lambda_{i, j+1},
		\blank
		\textup{(c)}\blank \hat\lambda_{ij} \geq \hat\lambda_{i+1, j},
		\blank\textup{and}\blank
		\textup{(d)}\blank \hat\lambda_{ij} \leq (\lambda_{0j}-q)^+ \textup{ if } i>0.
	\end{equation}
\end{theorem}

\section{State-Based Scheduling}\label{S:schedulability}

When each flow in our service model is guaranteed a worst-case service, we call the model a {\it \textbf{worst-case system}}. For all $\omega\in\Omega$, let $\bs\psi^\omega$ denote the worst-case service guaranteed to flow~$\omega$, and let $\bs\psi^{[\Omega]}$ denote the worst-case system. Since $\bs\psi^\omega$ is a state of flow~$\omega$, $\bs\psi^{[\Omega]}$ can be taken as the collective state of all flows in the system. Recall that $\bs\psi^\omega$ is conditioned on $b^\omega$, so $\bs\psi^{[\Omega]}$ is conditioned on $b^{[\Omega]}$. Accordingly, whenever we refer to $\bs\psi^{[\Omega]}$, we implicitly refer to $(\bs\psi^{[\Omega]}, b^{[\Omega]})$. Given $\bs\psi^{[\Omega]}$, can the server guarantee all flows their respective worst-case services simultaneously? If yes, how? To guarantee $\bs\psi^{[\Omega]}$, given any $a^{[\Omega]}$, a valid $d^{[\Omega]} \geq p^{[\Omega]}$ must be selected, where $d^{[\Omega]} \geq p^{[\Omega]}$ is the ensemble version of (\ref{E:WCSImmediate}). But this is not enough because $d^{[\Omega]}$ must also, via (\ref{E:WCSUpdateB}), induce a $\bs{\dot\psi}^{\raisebox{-0.9mm}{\scriptsize$[\Omega]$}}$ that can be guaranteed. This motivates the following definition.

\begin{definition}\label{D:schedulability}
	The \textbf{schedulability condition} is a condition on $\bs\psi^{[\Omega]}$ such that: (1) if it is not satisfied,  $\bs\psi^{[\Omega]}$ cannot be guaranteed; and (2) if it is satisfied, given any $a^{[\Omega]}$, there exists at least one valid $d^{[\Omega]} \geq p^{[\Omega]}$ that induces a  $\bs{\dot\psi}^{\raisebox{-0.9mm}{\scriptsize$[\Omega]$}}$ that satisfies the condition in the next slot.\footnote{Here (1) is necessary for the schedulability condition to be unique because otherwise, a schedulability condition for a server with capacity $c' < c$ would also be a schedulability condition for a server with capacity $c$.} We call such a $d^{[\Omega]}$, a \textbf{feasible schedule}, and a $\bs\psi^{[\Omega]}$ satisfying the schedulability condition, \textbf{schedulable}.
\end{definition}

Schedulability is the key to state-based scheduling because, once schedulable, a worst-case system can remain schedulable. Starting with a schedulable $\bs\psi^{[\Omega]}$, a state-based scheduler works iteratively. In each slot, it may first admit or deny new service requests provided that the reconfigured $\bs\psi^{[\Omega]}$ remains schedulable. Next, given any $a^{[\Omega]}$, it must select a feasible schedule, $d^{[\Omega]}$, to ensure that $\bs{\dot\psi}^{\raisebox{-0.9mm}{\scriptsize$[\Omega]$}}$ is schedulable in the next slot. It is illuminating to trace the scheduler's state evolution in its state space. In Fig.~\ref{F:StatePath}, a path from state~$A$ to $D$ is illustrated for a two-flow system. Notice that the states visited never leave the schedulable region bounded by the {\it Pareto} frontier, on which the service guaranteed to one flow must be reduced to improve that guaranteed to the other.

\begin{figure}[t]
	\centering \scalebox{0.875}{\includegraphics{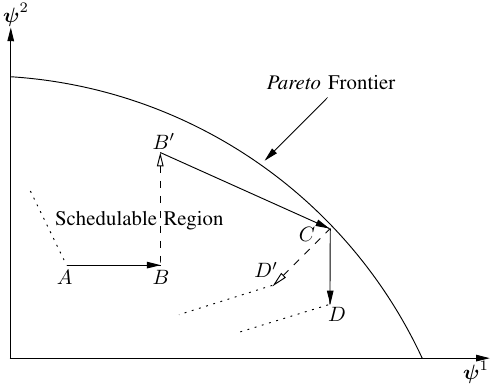}} \caption{In this two-flow case, the state-space path from state~$A$ to $D$ never leaves the schedulable region bounded by the {\it Pareto} frontier. The jump from state~$B$ to $B'$ reflects the admission of a new service request from flow~$2$. Had a different feasible schedule been selected in state~$C$, the path would have diverged to $D'$.}
	\label{F:StatePath}
	\Description{The schedulable region, the {\it Pareto} frontier, and a state path.}
\end{figure}

One clear advantage of state-based scheduling is that it is fully dynamic.\footnote{Although it is not the focus of the paper, this advantage is very attractive when worst-case services need to be negotiated on the fly. For instance, it allows on-the-fly adjustment of periodic or sporadic tasks' service parameters in real-time systems.} In each slot, $\bs\psi^{[\Omega]}$ can be reconfigured as long as it remains schedulable. For instance, in Fig.~\ref{F:StatePath}, the jump from state~$B$ to $B'$ reflects the admission of a new service request from flow~$2$. Another advantage is that, in each slot, a state-based scheduler has the flexibility to select any feasible schedule to serve the flows. Notice that different selections will induce different future states. For instance, in Fig.~\ref{F:StatePath}, had a different feasible schedule been selected in state~$C$, the path would have diverged to $D'$.

But what is the schedulability condition? Recall, from Section~\ref{SS:signatures}, that the least capacity that must be reserved during interval $[t+i, t+j)$ to guarantee worst-case service $\bs\psi^\omega$ to flow~$\omega$ is given by spectral value $\lambda_{ij}(\bs\psi^\omega)$. Intuitively, if $\bs\psi^{[\Omega]}$ can be guaranteed, then, during each $[t+i, t+j)$, the total capacity that must be reserved, $\sum_{\omega\in\Omega} \lambda_{ij}(\bs\psi^\omega)$, cannot exceed the server's available capacity, $(j-i)^+c$. If this reserve capacity is not available, $\bs\psi^{[\Omega]}$ cannot be guaranteed. This reserve capacity constraint turns out to be the schedulability condition that we seek.

\begin{theorem}\label{T:schedulability}
	A worst-case system, $\bs\psi^{[\Omega]}$, is schedulable if and only if its spectrum system satisfies the following condition:
	\begin{equation}\label{E:schedulability}
		\forall i,j\in\mb{N}, \blank\blank \lambda_{ij}^{\l\Omega\r} = \sum_{\omega\in\Omega} \lambda_{ij}(\bs\psi^\omega) \leq c_{ij} := (j-i)^+c.\footnote{This schedulability condition can be used as a normative formula to guide the design of multiplexing systems by measuring and comparing the capacity utilization of different multiplexing schemes. For details, see Appendix \ref{A:gains}.}
	\end{equation}
	A valid $d^{[\Omega]}$ is a feasible schedule if and only if the spectrum system that it induces in the next slot satisfies (\ref{E:schedulability}), that is,
	\begin{equation}\label{E:schedulability+}
		\forall i,j\in\mb{N}, \blank\blank \dot\lambda_{ij}^{\l\Omega\r} = \sum_{\omega\in\Omega} \lambda_{ij}(\bs{\dot\psi}^\omega) \leq \dot{c}_{ij} = c_{i+1, j+1}.
	\end{equation}
\end{theorem}

\section{Feasible Schedules}\label{S:feasible}

In this section, we identify all feasible schedules. It is in this regard that polymatroid theory comes to our attention. So we first give a primer on supermodular functions and permutohedra, both from polymatroid theory. We then introduce the baseline function and use it to show that the set of feasible schedules is a polytope that can be assembled from a sequence of permutohedral slices.

\subsection{Supermodular Functions and Permutohedra}\label{SS:polymatroid}

Polymatroid theory was first developed in \cite{Edmonds:1970}. An extensive survey can be found in \cite{Schrijver:2003} (ch. 44-49). Of particular interest are supermodular functions and permutohedra. Henceforth, for all $\Gamma,\Gamma'\subseteq\Omega$, we use $\Gamma+\Gamma'$ and $\Gamma\Gamma'$ to denote $\Gamma\cup\Gamma'$ and $\Gamma\cap\Gamma'$ respectively, and extending the notation of $x^{\l\Omega\r}$,  for all $\Gamma\subseteq\Omega$, we use $x^{\l\Gamma\r}$ to denote $\sum_{\omega\in\Gamma} x^\omega$, with $x^{\l\phi\r}:=0$.

\begin{definition}\label{D:supermodular}
	$\chi:2^\Omega \rightarrow \mb{N}$ is a \textbf{supermodular function} over $\Omega$ if
	\begin{equation}\label{E:supermodularEntries}
		\textup{(a)}\blank \chi(\phi) = 0,
		\blank\textup{and}\blank
		\textup{(b)}\blank  \forall \Gamma,\Gamma'\subseteq\Omega, ~\chi(\Gamma)+\chi(\Gamma') \leq \chi(\Gamma+\Gamma')+\chi(\Gamma\Gamma').\footnote{Our definition of supermodularity is more restrictive than is standard. But no generality is lost. We restrict the range of $\chi$ to $\mb{N}$, instead of $\mb{R}$, because our service model is discrete and non-negative. We require that $\chi(\phi) = 0$ because $\mb{P}(\chi)$ is empty when $\chi(\phi) > 0$.}
	\end{equation}
	If $\chi$ is supermodular, $\mb{P}(\chi)$ is the \textbf{permutohedron} generated by $\chi$ such that $d^{[\Omega]} \in \mb{P}(\chi)$ if and only if
	\begin{equation}\label{E:permutohedronEntries}
		\textup{(a)}\blank \forall \Gamma\subseteq\Omega, ~d^{\l\Gamma\r} \geq \chi(\Gamma),
		\blank\textup{and}\blank
		\textup{(b)}\blank d^{\l\Omega\r} = \chi(\Omega).
	\end{equation}
\end{definition}

By definition, $\mb{P}(\chi)$ is an $(n-1)$-polytope in general. Counted with multiplicity, it has $n!$ vertices that can be indexed by the $n!$ permutations over $\Omega$, which explains why such polytopes are called permutohedra. Permutohedra of dimensions $n=3$ and $n=4$ are illustrated in Fig.~\ref{F:Polymatroids}. A {\it \textbf{permutation}} over $\Omega$ is a bijective map, $\pi: \Omega \rightarrow \{1,2,\ldots,n\}$. Denote by $v_\pi^{[\Omega]}(\chi)$, the vertex of $\mb{P}(\chi)$ indexed by $\pi$. Then $v_\pi^{[\Omega]}(\chi)$ is the unique solution to the system of linear equations defined by
\begin{equation}\label{E:vertexSumEntries}
	\textup{(a)}\blank v_\pi^{\l\Gamma_\pi^i\r}(\chi) = \chi(\Gamma_\pi^i), ~i=0,1,2,\ldots,n,
	\blank\textup{with}\blank 
	\textup{(b)}\blank \Gamma_\pi^i := \{\omega\in\Omega | \pi(\omega) \leq i\},
\end{equation}
so for each $\omega\in\Omega$,
\begin{equation}\label{E:vertex}
	v_\pi^\omega(\chi) = \chi(\Gamma_\pi^{\pi(\omega)}) -\chi(\Gamma_\pi^{\pi(\omega)-1}).
\end{equation}

\begin{figure}[t]
	\centering \scalebox{0.625}{\includegraphics{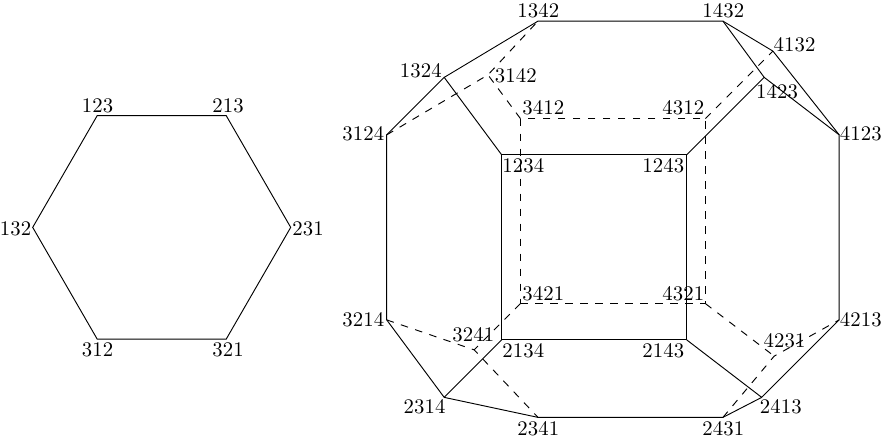}} \caption{When $n=3$, the permutohedron is a hexagon with $6$ vertices and $6$ edges. When $n=4$, it is a truncated octahedron with $24$ vertices, $36$ edges and $14$ facets, among which $6$ are rectangles and $8$ are hexagons.} \label{F:Polymatroids}
	\Description{Permutohedra of orders $n=3$ and $n=4$.}
\end{figure}

\subsection{The Feasible Polytope and Feasible Permutohedra}\label{SS:FP}

According to Theorem~\ref{T:schedulability}, if $d^{[\Omega]}$ is feasible, it must ensure that $\dot\lambda_{ij}^{\l\Omega\r} \leq c_{i+1, j+1}$ for all $i,j \in \mb{N}$. If $\bs\psi^{[\Omega]}$ is schedulable, this requirement is satisfied by default in the case that $i > 0$, because using (\ref{E:SigUpdateCaseB}), (\ref{E:CondSigPropF}), and (\ref{E:schedulability}), we have $\dot\lambda_{ij}^{\l\Omega\r} = \hat\lambda_{i+1, j+1}^{\l\Omega\r} \leq \lambda_{i+1, j+1}^{\l\Omega\r} \leq c_{i+1, j+1}$. So we need only focus on the case that $i=0$. In this case, given any $d^{[\Omega]} \leq q^{[\Omega]}$, according to (\ref{E:SigUpdateCaseA}) and (\ref{E:lambdaBound}), $\dot\lambda_{0j}^{\l\Omega\r} \leq c_{1, j+1}$ implies that, for all $\Gamma\subseteq\Omega$.
\[
	c_{1, j+1} \geq  \dot\lambda_{0j}^{\l\Omega\r} = \dot\lambda_{0j}^{\l\Gamma\r}+\dot\lambda_{0j}^{\l\overline\Gamma\r} = \sum_{\omega\in\Gamma} (\hat\lambda_{0, j+1}^\omega-d^\omega)^+ + \dot\lambda_{0j}^{\l\overline\Gamma\r} \geq \hat\lambda_{0, j+1}^{\l\Gamma\r}-d^{\l\Gamma\r} + \hat\lambda_{1, j+1}^{\l\overline\Gamma\r},
\]
where $\overline\Gamma$ denotes $\Omega\setminus\Gamma$. This in turn implies that
\begin{equation}\label{E:beta}
	 \forall \Gamma\subseteq\Omega, \blank\blank d^{\l\Gamma\r} \geq \beta(\Gamma) := \max_{j \in \mb{N}} (\hat\lambda_{0, j+1}^{\l\Gamma\r} + \hat\lambda_{1, j+1}^{\l\overline\Gamma\r}-c_{1, j+1}),
\end{equation}
As (\ref{E:beta}) specifies the least number of tasks that must be served from any subset of flows to guarantee $\bs\psi^{[\Omega]}$, we call it the {\it \textbf{baseline constraint}} and $\beta$ the {\it \textbf{baseline function}}. The importance of $\beta$ is highlighted in the next theorem.

\begin{theorem}\label{T:polytope}
	If $\bs\psi^{[\Omega]}$ is schedulable, a valid $d^{[\Omega]}$ is a feasible schedule if and only if it satisfies the baseline constraint, (\ref{E:beta}).
\end{theorem}

According to this theorem, if $\bs\psi^{[\Omega]}$ is schedulable, the set of feasible schedules is completely determined by (\ref{E:causalityB}), (\ref{E:capacity}), and (\ref{E:beta}), all linear constraints. So it is an $n$-polytope in general. We call it the {\it \textbf{feasible polytope}} and denote it by $\mb{F}$. In the two-flow case illustrated in Fig.~\ref{F:FeasibleRegion}, it is the hexagon, $ABCDEF$. In three-flow cases, it resembles a diamond. But, in more general cases, what is $\mb{F}$'s structure? To derive its structure, we need a key property of $\beta$.

\begin{figure}[t]
	\centering \scalebox{0.875}{\includegraphics{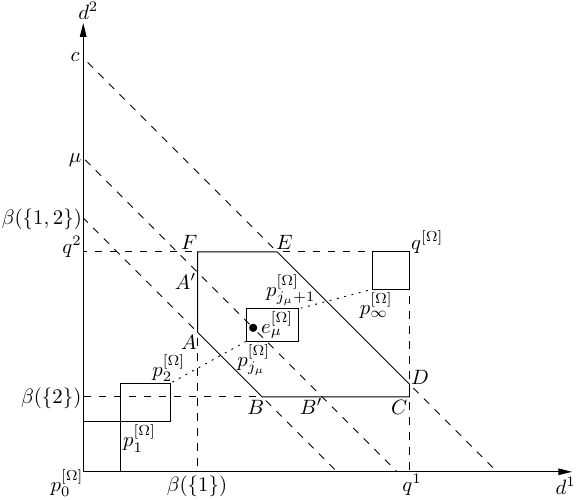}} \caption{In this two-flow case, the feasible polytope is the hexagon, $ABCDEF$. Its intersection with $\mb{H}_\mu$ is a feasible permutohedron, line segment $\overline{A'B'}$. Max-slack schedule $e_\mu^{[\Omega]}$ lies in the intersection of $\mb{H}_\mu$ and the elevating staircase of rectangles, $[p_j^{[\Omega]}, p_{j+1}^{[\Omega]}]$, $j = 0, 1, 2, \ldots$, ending with $[p_\infty^{[\Omega]}, q^{[\Omega]}]$.}
	\label{F:FeasibleRegion}
	\Description{The feasible polytope, a feasible permutohedron, and a max-slack schedule.}
\end{figure}

\begin{theorem}\label{T:betaProp}
	If $\bs\psi^{[\Omega]}$ is schedulable, $\beta$ is supermodular.
\end{theorem}

Given $\beta$'s supermodularity, it would seem that $\mb{F}$ should be related to some permutohedron. The problem is that $d^{\l\Omega\r}$ must remain constant in a permutohedron, which is not the case for $\mb{F}$. This observation motivates us to intersect $\mb{F}$ with the hyperplane,
$\mb{H}_\mu$, defined by $d^{\l\Omega\r} = \mu$, as the intersection, $\mb{F}_\mu:=\mb{F}\cap\mb{H}_\mu$, turns out to be a permutohedron.

\subsubsection*{\textup{[}Feasible Permutohedron\textup{]}}When $d^{[\Omega]} \in \mb{F}_\mu$, it must satisfy (\ref{E:causalityB}), (\ref{E:capacity}), and (\ref{E:beta}). Hence, for $\mb{F}_\mu$ to be non-empty, $\mu$ must satisfy
\begin{equation}\label{E:muFeasible}
	\beta(\Omega) \leq \mu \leq \min \{c, q^{\l\Omega\r}\}.
\end{equation}
We call such a $\mu$ {\it \textbf{feasible}}. If $\bs\psi^{[\Omega]}$ is schedulable, at least one feasible $\mu$ exists. To see this, on the one hand, using (\ref{E:beta}), (\ref{E:CondSigPropF}), and (\ref{E:schedulability}), we have
\[
	\beta(\Omega) = \max_{j \in \mb{N}} (\hat\lambda_{0, j+1}^{\l\Omega\r} -c_{1, j+1}) \leq \max_{j \in \mb{N}} (\lambda_{0, j+1}^{\l\Omega\r} -c_{1, j+1}) \leq \max_{j \in \mb{N}} (c_{0, j+1} -c_{1, j+1}) = c.
\]
On the other hand, according to (\ref{E:CondSigPropE}) and (\ref{E:CondSigPropF}), $\hat\lambda_{0, j+1} \leq q+\hat\lambda_{1, j+1} \leq q+\lambda_{1, j+1}$, so using (\ref{E:beta}) and (\ref{E:schedulability}), we also have
\[
	\beta(\Omega) = \max_{j \in \mb{N}} (\hat\lambda_{0, j+1}^{\l\Omega\r} -c_{1, j+1}) \leq \max_{j \in \mb{N}} (q^{\l\Omega\r}+\lambda_{1, j+1}^{\l\Omega\r} -c_{1, j+1}) \leq q^{\l\Omega\r} + \max_{j \in \mb{N}} (c_{1, j+1} -c_{1, j+1}) = q^{\l\Omega\r}.
\]

For clarity, when $d^{[\Omega]} \in \mb{H}_\mu$, we denote it by $d_\mu^{[\Omega]}$ so that, by definition, $d_\mu^{\l\Omega\r} = \mu$. Then, when $d_\mu^{[\Omega]} \in \mb{F}_\mu$, since (\ref{E:beta}) and (\ref{E:causalityB}) imply that $d_\mu^{\l\Gamma\r} \geq \beta(\Gamma)$ and $d_\mu^{\l\Gamma\r}+q^{\l\overline\Gamma\r} \geq d_\mu^{\l\Omega\r} = \mu$,
\begin{equation}\label{E:betaMu}
	\forall \Gamma\subseteq\Omega, \blank\blank d_\mu^{\l\Gamma\r} \geq \beta_\mu(\Gamma) := \max\{\beta(\Gamma), \mu-q^{\l\overline\Gamma\r}\}.
\end{equation}
If $\mu$ is feasible, it is also immediate from (\ref{E:muFeasible}) that $\beta_\mu(\Omega) = \max\{\beta(\Omega), \mu\} = \mu$. So, when $d_\mu^{[\Omega]} \in \mb{F}_\mu$, we have both
\begin{equation}\label{E:permuBetaMuEntries}
	\text{(a)}\blank \forall \Gamma\subseteq\Omega, ~d_\mu^{\l\Gamma\r} \geq \beta_\mu(\Gamma),
	\blank\text{and}\blank
	\text{(b)}\blank d_\mu^{\l\Omega\r} = \mu = \beta_\mu(\Omega).
\end{equation}
Comparing (\ref{E:permuBetaMuEntries}) to (\ref{E:permutohedronEntries}), it is clear that, if $\beta_\mu$ is supermodular, $d_\mu^{[\Omega]} \in \mb{P}(\beta_\mu)$, and thus $\mb{F}_\mu \subseteq \mb{P}(\beta_\mu)$. But is $\beta_\mu$ supermodular? Yes, according to the next theorem, it is, and moreover, $\mb{F}_\mu$ is exactly $\mb{P}(\beta_\mu)$.

\begin{theorem}\label{T:FP}
	If $\bs\psi^{[\Omega]}$ is schedulable and $\mu$ is feasible, $\beta_\mu$ is supermodular, and $\mb{F}_\mu = \mb{P}(\beta_\mu)$.
\end{theorem}

Since $\mb{P}(\beta_\mu)$ is non-empty, this theorem establishes that, if $\bs\psi^{[\Omega]}$ is schedulable, given any $a^{[\Omega]}$, at least one feasible schedule exists.\footnote{Notice that $\beta_\mu$ implicitly depends on $a^{[\Omega]}$. This can be seen in the following chain of dependence: $\beta_\mu$ on $\beta$ via (\ref{E:betaMu}), $\beta$ on $\hat\lambda_{ij}^{[\Omega]}$ via (\ref{E:beta}), $\hat\lambda_{ij}^{[\Omega]}$ on $q^{[\Omega]}$ via (\ref{E:CondSig}), and finally $q^{[\Omega]}$ on $a^{[\Omega]}$ via (\ref{E:causalityB}).} We call $\mb{P}(\beta_\mu)$ the {\it \textbf{feasible permutohedron}} under $\mu$. When $n=2$, it is a line segment, which is $\overline{A'B'}$ in Fig.~\ref{F:FeasibleRegion}. When $n=3$, it is a hexagon. When $n=4$, it is a truncated octahedron. These two cases are depicted in Fig.~\ref{F:Polymatroids}. Although each $\mb{P}(\beta_\mu)$ is but one slice of $\mb{F}$, by putting all of the slices together, we can assemble the entire $\mb{F}$.

According to Theorem~\ref{T:FP}, to select a feasible schedule, we can first select a feasible $\mu$, which fixes the total service, and then select $d_\mu^{[\Omega]}$ from $\mb{P}(\beta_\mu)$.\footnote{For an alternative, and complementary, approach to selecting feasible schedules, see Appendix \ref{A:permuBase}.} For instance, if a vertex of $\mb{P}(\beta_\mu)$, $v_\pi^{[\Omega]}(\beta_\mu)$, is selected, according to (\ref{E:vertexSumEntries})-(a), $v_\pi^{\l\Gamma_\pi^i\r}(\beta_\mu) = \beta_\mu(\Gamma_\pi^i)$, $i=0,1,2,\ldots,n$. Comparing this to (\ref{E:betaMu}), it is immediate that  $d_\mu^{\l\Gamma_\pi^i\r}$ is minimized by this selection. So a strict flow priority order is enforced by selecting $v_\pi^{[\Omega]}(\beta_\mu)$. In particular,  according to (\ref{E:vertexSumEntries})-(b), $\pi(\omega)$ can be viewed as a priority index such that the larger the $\pi(\omega)$, the higher the priority that flow $\omega$ enjoys. When, instead of enforcing priorities, the objective is ensuring fairness, priorities can be assigned equally by selecting the vertex centroid of $\mb{P}(\beta_\mu)$,
\begin{equation}\label{E:fair}
	v_\text{F}^{[\Omega]}(\beta_\mu):= \frac{1}{n!} \sum_{\pi\in\Pi^\Omega} v_\pi^{[\Omega]}(\beta_\mu),
\end{equation}
where $\Pi^\Omega$ is the set of all permutations over $\Omega$. Of course, (\ref{E:fair}) may not result in an integral point, in which case, it can be rounded. Interestingly, $v_\text{F}^{[\Omega]}(\beta_\mu)$'s definition coincides with that of the {\it Shapley} value from cooperative game theory \cite{Shapley:1953}.

\section{Max-Slack Schedules}\label{S:MSS}

A downside of our framework's generality is its complexity. One source of this complexity is the fact that, to fully exploit the flexibility of selecting any $d_\mu^{[\Omega]}$ from $\mb{P}(\beta_\mu)$, such as (\ref{E:fair}), we must calculate all $2^n$ values of $\beta_\mu$. However, for a special class of schedules, max-slack schedules, all of this calculation can be avoided. In the remainder of this section, we first define max-slack schedules and investigate their properties. We then show how per-class max-slack schedules can be constructed to enable intermediate tradeoffs between flexibility and efficiency.

\subsection{Definition and Properties}\label{SS:MSSDef}

Let $\bs{p} = [p_j]_{j \in \mb{N}} \in \mb{U}$, with
\begin{equation}\label{E:p+}
	p_j := \min \{\hat\lambda_{0j}, q\}.
\end{equation}
It is immediate from (\ref{E:CondSig}) that $p_j =  \max_{\bs{q} \in \mb{U}|q} \min \{\psi_j(\bs{q}), q\}$. That is to say, to guarantee $\bs{d} \geq \bs\psi(\bs{q})$ no matter which $\bs{q} \in \mb{U}|q$ is realized, at least $p_j$ of the $q$ tasks queued in the buffer must be served during interval $[t, t+j)$. According to (\ref{E:p+}) and (\ref{E:CondSigPropEntries})-(a, b), $\bs{p}$ is a cumulative vector. It is also a natural extension of $p$ because according to (\ref{E:p+}), (\ref{E:p-lambda}), and (\ref{E:p<q}), $p_1 = \min\{\hat\lambda_{01}, q\} = \min\{p, q\} = p$. Significantly, the condition for $p_j^{[\Omega]}$ to be feasible is especially simple.

\begin{theorem}\label{L:pFeasible}
{\it If $\bs\psi^{[\Omega]}$ is schedulable, $p_j^{[\Omega]}$ is a feasible schedule if $\beta(\Omega) \leq p_j^{\l\Omega\r} \leq c$.}	
\end{theorem}

As illustrated in Fig.~\ref{F:FeasibleRegion}, $p_j^{[\Omega]}$'s are disjoint points separated by irregular gaps. When $\mb{H}_\mu$ intersects $[p_{j_\mu}^{[\Omega]}, p_{j_\mu+1}^{[\Omega]}]$, the intersection lies between $p_{j_\mu}^{[\Omega]}$ and $p_{j_\mu+1}^{[\Omega]}$. This motivates the following definition.

\begin{definition}\label{D:MSS}
	$e_\mu^{[\Omega]} \in \mb{H}_\mu$ is a \textbf{max-slack schedule} under~$\mu$~if
	\begin{equation}\label{E:MSSDefEntries}
		\textup{(a)}\blank e_\mu^{[\Omega]}\in[p_{j_\mu}^{[\Omega]}, p_{j_\mu+1}^{[\Omega]}] \textup{ if } j_\mu < \infty,
		\blank\textup{and}\blank
		\textup{(b)}\blank e_\mu^{[\Omega]}\in[p_\infty^{[\Omega]}, q^{[\Omega]}] \textup{ if } j_\mu = \infty,
	\end{equation}
	where
	\begin{equation}\label{E:jMu}
		j_\mu := \tau_{\mu+1}(\bs{p}^{\l\Omega\r}).
	\end{equation}
\end{definition}

Intuitively, $e_\mu^{[\Omega]}$ lies in the intersection of $\mb{H}_\mu$ and the elevating staircase of hypercuboids, $[p_j^{[\Omega]}, p_{j+1}^{[\Omega]}]$, $j = 0, 1, 2, \ldots$, ending with $[p_\infty^{[\Omega]}, q^{[\Omega]}]$, which reduces to a staircase of rectangles in Fig.~\ref{F:FeasibleRegion}. The introduction of the last hypercuboid, $[p_\infty^{[\Omega]}, q^{[\Omega]}]$, ensures that $e_\mu^{[\Omega]} \leq q^{[\Omega]}$. Clearly, if $\mu > q^{\l\Omega\r}$, no intersection is possible, so no $e_\mu^{[\Omega]}$ exists. Otherwise, the location of the intersection is determined by $j_\mu$. To see why, notice that, if $\mu \leq q^{\l\Omega\r}$, using (\ref{E:jMu}) and (\ref{E:tau}), we have
\begin{equation}\label{E:jMu+Entries}
	\textup{(a)}\blank \mu\in[p_{j_\mu}^{\l\Omega\r}, p_{j_\mu+1}^{\l\Omega\r}) \textup{ if } j_\mu < \infty,
	\blank\textup{and}\blank
	\textup{(b)}\blank \mu\in[p_\infty^{\l\Omega\r}, q^{\l\Omega\r}] \text{ if } j_\mu = \infty.
\end{equation}
So $\mb{H}_\mu$ intersects the staircase somewhere in $[p_{j_\mu}^{[\Omega]}, p_{j_\mu+1}^{[\Omega]}]$ if $j_\mu < \infty$, and in $[p_\infty^{[\Omega]}, q^{[\Omega]}]$ if $j_\mu = \infty$.

\subsubsection*{\textup{[}The Properties of \texorpdfstring{$e_\mu^{[\Omega]}$\textup{]}}} Let $\mu = p_j^{\l\Omega\r}$. If $\bs\psi^{[\Omega]}$ is schedulable, according to (\ref{E:muFeasible}), and Theorem~\ref{L:pFeasible}, $p_j^{[\Omega]} \in \mb{F}_\mu$ if $\mu$ is feasible. Recall that $\mu$'s feasibility is necessary for $\mb{F}_\mu$ to be non-empty. So, given $\mu = p_j^{\l\Omega\r}$, $p_j^{[\Omega]} \in \mb{F}_\mu$ if $\mb{F}_\mu$ is non-empty. We further conjecture that, given any $\mu$, $e_\mu^{[\Omega]} \in \mb{F}_\mu$ if $\mb{F}_\mu$ is non-empty. This conjecture is confirmed by the next theorem.

\begin{theorem}\label{T:MSS}
	{If~} $\mb{F}_\mu$ is non-empty, $e_\mu^{[\Omega]}$ exists, and $e_\mu^{[\Omega]} \in \mb{F}_\mu$.
\end{theorem}

According to Definition \ref{D:MSS}, $e_\mu^{[\Omega]}$ is confined to a hypercuboid that can be easily determined from $\mu$, $\bs{p}^{[\Omega]}$, and $q^{[\Omega]}$ alone. According to Theorems~\ref{T:FP} and \ref{T:MSS}, $e_\mu^{[\Omega]}$ is feasible if $\bs\psi^{[\Omega]}$ is schedulable and $\mu$ is feasible. So, to construct a feasible $e_\mu^{[\Omega]}$, no calculation of $\beta_\mu$ is required at all. Of course, to ensure that $\mu$ is feasible, according to (\ref{E:muFeasible}), we still need to calculate $\beta(\Omega)$. But even this calculation can be avoided if we let $\mu=\min\{c, q^{\l\Omega\r}\}$, that is, if we require that the server be work-conserving. Underlying all of these simplifications is the fact that the shape of $\mb{F}_\mu$ is irrelevant in Theorem~\ref{T:MSS}. This is explained by the next theorem.

\begin{theorem}\label{L:MSS}
	If $e_\mu^{[\Omega]}$ exists, for all $i,j\in\mb{N}$
	\begin{equation}\label{E:MSS}
		\dot\lambda_{ij}^{\l\Omega\r}(e_\mu^{[\Omega]}) = \min_{d_\mu^{[\Omega]} \leq q^{[\Omega]}} \dot\lambda_{ij}^{\l\Omega\r}(d_\mu^{[\Omega]}).
	\end{equation}
\end{theorem}

Notice that, in (\ref{E:MSS}), $\dot\lambda_{ij}^{\l\Omega\r}$'s dependence on $e_\mu^{[\Omega]}$ or $d_\mu^{[\Omega]}$ through (\ref{E:SigUpdateCases}) is spelled out explicitly. According to Theorems~\ref{T:MSS} and \ref{L:MSS}, if $\mb{F}_\mu$ is non-empty, $e_\mu^{[\Omega]}$ is not just any feasible schedule in $\mb{F}_\mu$ but one that minimizes $\dot\lambda_{ij}^{\l\Omega\r}$ for all $i,j \in \mb{N}$ simultaneously. According to Theorem~\ref{T:schedulability}, this implies that $e_\mu^{[\Omega]}$ leaves maximum room for the server to admit new service requests in the next slot, which explains why we call it the max-slack schedule. Moreover, as $e_\mu^{[\Omega]}$ minimizes $\dot\lambda_{ij}^{\l\Omega\r}$ for all $i,j \in \mb{N}$ simultaneously, if any valid $d_\mu^{[\Omega]}$ can keep $\dot\lambda_{ij}^{\l\Omega\r} \leq \dot{c}_{ij}$ for all $i,j \in \mb{N}$, so can $e_\mu^{[\Omega]}$. Thus, according to Theorem~\ref{T:schedulability}, as long as $\mb{F}_\mu$ is non-empty, that is, as long as any $d_\mu^{[\Omega]}$ is feasible, so is $e_\mu^{[\Omega]}$. This explains not only why Theorem~\ref{T:MSS} holds but also why the shape of $\mb{F}_\mu$ is irrelevant in Theorem~\ref{T:MSS}. 

\subsection{Per-Class Max-Slack Schedules}\label{SS:HMSS}

By selecting $e_\mu^{[\Omega]}$ to avoid all calculation of $\beta_\mu$, we lose almost all flexibility. To explore the full flexibility-efficiency continuum instead of being limited to either extreme, we can aggregate flows into classes. Let $\mc{P} \subseteq 2^\Omega$ be a {\it \textbf{partition}} of $\Omega$ so that $\bigcup_{\Gamma\in\mc{P}} \Gamma = \Omega$ and $\Gamma\Gamma'=\phi$ for all distinct $\Gamma,\Gamma'\in\mc{P}$. When $\Gamma \in \mc{P}$ is used as an index, as in $\nu^\Gamma$, we denote the ensemble of all $\nu^\Gamma$'s by $\nu^{[\mc{P}]}$, and the sum, $\sum_{\Gamma\in\mc{P}} \nu^\Gamma$, by $\nu^{\l\mc{P}\r}$. We call $\nu^{[\mc{P}]}$ an {\it \textbf{inter-class schedule}} if, for each $\Gamma\in\mc{P}$, $\nu^\Gamma$ is the total service to flows in $\Gamma$ so that, flow-by-flow, schedules are restricted to
\begin{equation}\label{E:Hnu}
	\mb{H}(\nu^{[\mc{P}]}) := \{d^{[\Omega]} | d^{\l\Gamma\r} = \nu^\Gamma ~\forall \Gamma\in\mc{P}\}.
\end{equation}
Let $\mu = \nu^{\l\mc{P}\r}$, and for clarity, denote $\nu^{[\mc{P}]}$ by $\nu_\mu^{[\mc{P}]}$. It is immediate that $\mb{H}(\nu_\mu^{[\mc{P}]}) \subseteq \mb{H}_\mu$. Additionally, let $\mb{F}(\nu_\mu^{[\mc{P}]}) := \mb{F} \cap \mb{H}(\nu_\mu^{[\mc{P}]})$. So $\mb{F}(\nu_\mu^{[\mc{P}]}) \subseteq \mb{F}_\mu$.

Roughly speaking, given $\nu_\mu^{[\mc{P}]}$, for each $\Gamma\in\mc{P}$, we can construct a max-slack schedule under $\nu_\mu^\Gamma$ for $\bs\psi^{[\Gamma]}$, where $\bs\psi^{[\Gamma]}$ denotes the subsystem of $\bs\psi^{[\Omega]}$ indexed by $\Gamma$. By concatenating all such schedules, flows are, inter-class, scheduled according to $\nu_\mu^{[\mc{P}]}$ and, intra-class, max-slack scheduled. This motivates the following extensions of Definition~\ref{D:MSS} and Theorem~\ref{T:MSS}.

\begin{definition}\label{D:HMSS}
	$e^{[\Omega]}(\nu_\mu^{[\mc{P}]}) \in \mb{H}(\nu_\mu^{[\mc{P}]})$ is a \textbf{per-class max-slack schedule} under $\nu_\mu^{[\mc{P}]}$ if, for each $\Gamma\in\mc{P}$,
	\begin{equation}\label{E:HMSSDefEntries}
		\textup{(a)}\blank e^{[\Gamma]}(\nu_\mu^{[\mc{P}]}) \in	[p_{j_\mu^\Gamma}^{[\Gamma]},p_{j_\mu^\Gamma+1}^{[\Gamma]}] \textup{ if } j_\mu^\Gamma < \infty, \blank\textup{and}\blank
		\textup{(b)}\blank e^{[\Gamma]}(\nu_\mu^{[\mc{P}]}) \in	[p_\infty^{[\Gamma]} , q^{[\Gamma]}] \textup{ if } j_\mu^\Gamma = \infty,
	\end{equation}
where
\begin{equation}\label{E:jT}
	j_\mu^\Gamma := \tau_{\nu_\mu^\Gamma+1}(\bs{p}^{\l\Gamma\r}).
\end{equation}
\end{definition}

\begin{theorem}\label{T:HMSS}
	{If~} $\mb{F}(\nu_\mu^{[\mc{P}]})$ is non-empty, $e^{[\Omega]}(\nu_\mu^{[\mc{P}]})$ exists, and $e^{[\Omega]}(\nu_\mu^{[\mc{P}]}) \in \mb{F}(\nu_\mu^{[\mc{P}]})$.
\end{theorem}

But how do we know if $\mb{F}(\nu_\mu^{[\mc{P}]})$ is non-empty?
By definition, $\mb{F}(\nu_\mu^{[\mc{P}]}) \subseteq \mb{F}_\mu$ while, according to Theorem~\ref{T:FP}, $\mb{F}_\mu = \mb{P}(\beta_\mu)$. So, if $\mb{F}(\nu_\mu^{[\mc{P}]})$ is non-empty, given any $d_\mu^{[\Omega]} \in \mb{F}(\nu_\mu^{[\mc{P}]})$, $d_\mu^{[\Omega]} \in \mb{P}(\beta_\mu)$, that is, $d_\mu^{[\Omega]}$ must satisfy (\ref{E:permuBetaMuEntries}). This implies that, for all $\mc{S} \subseteq \mc{P}$, $
\nu_\mu^{\l\mc{S}\r} = \sum_{\Gamma\in\mc{S}} d_\mu^{\l\Gamma\r} = d_\mu^{\l\Gamma_\mc{S}\r} \geq \beta_\mu(\Gamma_\mc{S})$, where $\nu_\mu^{\l\mc{S}\r}$ denotes $\sum_{\Gamma\in\mc{S}} \nu_\mu^\Gamma$ and $\Gamma_\mc{S}$ denotes $\bigcup_{\Gamma\in\mc{S}} \Gamma$, and that $\nu_\mu^{\l\mc{P}\r} = d_\mu^{\l\Omega\r} = \mu = \beta_\mu(\Omega)$. Therefore,
\begin{equation}\label{E:permuBetaPEntries}
	\text{(a)}\blank  \forall \mc{S} \subseteq \mc{P}, ~\nu_\mu^{\l\mc{S}\r} \geq \beta_\mu^\mc{P}(\mc{S}) :=\beta_\mu(\Gamma_\mc{S}),
	\blank\text{and}\blank	
	\text{(b)}\blank \nu_\mu^{\l\mc{P}\r} = \mu = \beta_\mu(\Omega) = \beta_\mu^\mc{P}(\mc{P}).
\end{equation}
As the sampling of $\beta_\mu$ on the algebra generated by $\mc{P}$, $\beta_\mu^\mc{P}$ is supermodular. We call $\nu_\mu^{[\mc{P}]}$ {\it \textbf{feasible}} if $\mb{F}(\nu_\mu^{[\mc{P}]})$ is non-empty. So, comparing (\ref{E:permuBetaPEntries}) to (\ref{E:permutohedronEntries}), it is immediate that a necessary condition for $\nu_\mu^{[\mc{P}]}$ to be feasible is that $\nu_\mu^{[\mc{P}]} \in \mb{P}(\beta_\mu^\mc{P})$. It turns out that this condition is also sufficient.

\begin{theorem}\label{T:HMSSCond}
	If $\bs\psi^{[\Omega]}$ is schedulable and $\mu$ is feasible, $\nu_\mu^{[\mc{P}]}$ is feasible if and only if $\nu_\mu^{[\mc{P}]} \in \mb{P}(\beta_\mu^\mc{P})$.
\end{theorem}

According to Theorems~\ref{T:HMSS} and \ref{T:HMSSCond}, to select a per-class max-slack schedule, we can first select a feasible $\mu$, then select a feasible $\nu_\mu^{[\mc{P}]}$ from $\mb{P}(\beta_\mu^\mc{P})$, and finally construct $e^{[\Omega]}(\nu_\mu^{[\mc{P}]})$, which, according to Definition~\ref{D:HMSS}, is confined to a hypercuboid that can be easily determined from $\nu_\mu^{[\mc{P}]}$, $\bs{p}^{[\Omega]}$ and $q^{[\Omega]}$ alone. This enables intermediate tradeoffs between flexibility and efficiency because, as $|\mc{P}|$ increases or decreases, $\mb{P}(\beta_\mu^\mc{P})$ becomes harder or easier to determine, but we gain or lose flexibility.\footnote{To see how per-class max-slack schedules can be selected to enforce priorities or ensure fairness, see Appendix \ref{A:selectPerClass}.}

\section{Min-Plus Services}\label{S:MPS}

A second source of complexity in our framework is the fact that worst-case services, as uncountably infinite, full-blown maps between cumulative vectors, are difficult to specify and update. In this section, we show that, among all worst-case services that share the same spectrum, there exists a maximum and this maximum can be constructed using the min-plus algebra.\footnote{A given spectrum is almost always shared by multiple worst-case services because, as crude counts show, the cardinal of the set of all worst-cases is $\aleph_2$ while the cardinal of the set of all spectra is $\aleph_1$.} So, by construction, min-plus services can be completely identified by their spectra, which are countably infinite. We introduce two definitions for min-plus services. The first is convenient for identifying their spectra, and the second, for identifying their update rule.

\subsection{A First Definition and the Spectrum}\label{SS:MPSDef}

Although, in general, worst-case services cannot be identified by their spectra completely, they can be upper bounded by their spectra. According to (\ref{E:signature}), for all $\bs{q} \in \mb{U} \ua b$ and $i,j\in\mb{N}$, $\psi_j(\bs{q}) -q_i \leq \lambda_{ij}(\bs\psi)$, so
\begin{equation}\label{E:psiBound}
	\psi_j(\bs{q}) \leq \min_{i \in \mb{N}} (q_i+\lambda_{ij}(\bs\psi)).
\end{equation}
This bound can be rewritten concisely using the {\it \textbf{min-plus algebra}}, an algebra in which operators $\min$ and $+$ replace, respectively, operators $+$ and $\times$ in the standard algebra. First, construct the semi-infinite matrix of $\lambda_{ij}(\bs\psi)$'s,
\[
	\Lambda(\bs\psi) = [\lambda_{ij}(\bs\psi)]_{i,j\in\mb{N}} = \left[
	\begin{array}{ccccc}
		\lambda_{00}(\bs\psi) & \lambda_{01}(\bs\psi) & \lambda_{02}(\bs\psi) & \cdots\\
		\lambda_{10}(\bs\psi) & \lambda_{11}(\bs\psi) & \lambda_{12}(\bs\psi) & \cdots\\
		\lambda_{20}(\bs\psi) & \lambda_{21}(\bs\psi) & 	\lambda_{22}(\bs\psi) & \cdots\\
		\vdots                & \vdots                & \vdots                & \ddots
	\end{array}
	\right].
\]
Next, using $\otimes$ to denote min-plus matrix multiplication, and applying the matrix multiplication rule, rewrite (\ref{E:psiBound}) in matrix form as
\begin{equation}\label{E:psiBound+}
	\bs\psi(\bs{q}) \leq \bs{q} \otimes \Lambda(\bs\psi) \blank\blank \forall \bs{q} \in \mb{U} \ua b.
\end{equation}
This version of the bound motivates the following definition.

\begin{definition}\label{D:SM-MPS}
	For a flow with $b$ tasks left unserved prior to slot $t$, $S=[s_{ij}]_{i,j\in\mb{N}}$ is a \textbf{spectral matrix} if, for all $i,j\in\mb{N}$,
	\begin{equation}\label{E:SMDefEntries}
		\textup{(a)}\blank s_{ij} = 0 \textup{ if } i \geq j,
		\blank
		\textup{(b)}\blank s_{ij} \leq s_{i, j+1},
		\blank
		\textup{(c)}\blank s_{ij} \geq s_{i+1, j},
		\blank\textup{and}\blank
		\textup{(d)}\blank s_{ij} \leq (s_{0j}-b)^+ \textup{ if } i>0.
	\end{equation}
	We call $\bs\psi^S$ the \textbf{min-plus service} identified by $S$ if
	\begin{equation}\label{E:psiSVec}
		\bs\psi^S(\bs{q}):= \bs{q} \otimes S \blank\blank \forall \bs{q} \in \mb{U}\ua b.
	\end{equation}
\end{definition}

According to this definition, spectral matrices are non-negative, strictly upper triangular, and have non-decreasing rows and non-increasing columns. It also follows from (\ref{E:psiSVec}) that, for all $j\in\mb{N}$,
\begin{equation}\label{E:psiS}
	\psi_j^S(\bs{q}) = \bs{q} \otimes S_{\bs\cdot j} = \min_{i\in\mb{N}} (q_i+s_{ij}) = \min_{i \leq j} (q_i+s_{ij}),
\end{equation}
where $S_{\bs\cdot j}$ denotes the $j$th column of $S$, and the last equality holds because, according to (\ref{E:SMDefEntries})-(a),
\[
	\min_{i\in\mb{N}} (q_i+s_{ij}) = \min \left\{\min_{i<j} (q_i+s_{ij}), \min_{i \geq j} q_i\right\} = \min \left\{\min_{i<j} (q_i+s_{ij}), q_j\right\} = \min_{i \leq j} (q_i+s_{ij}).
\]
This derivation shows that $\psi_j^S(\bs{q}) \leq q_j$. So $\bs\psi^S(\bs{q})\leq \bs{q}$. Moreover, according to (\ref{E:psiS}) and (\ref{E:SMDefEntries})-(b), $\psi_j^S(\bs{q}) \leq \psi_{j+1}^S(\bs{q})$. So $\bs\psi^S(\bs{q})$ is a cumulative vector. Therefore, according to Definition~\ref{D:WCS}, $\bs\psi^S$ is indeed a worst-case service.

\subsubsection*{\textup{[}The Spectrum\textup{]}} By construction, (\ref{E:SMDefEntries}) parallels (\ref{E:SigPropEntries}). So, not surprisingly, $s_{ij}$ behaves very much like $\lambda_{ij}$, the spectral value of {\it some} worst-case service. It is fundamental that {\it this} worst-case service can be $\bs\psi^S$ itself.

\begin{theorem}\label{T:SSig}
	For all $i,j\in\mb{N}$,
	\begin{equation}\label{E:SSig}
		\lambda_{ij}(\bs\psi^S) = s_{ij}.
	\end{equation}
\end{theorem}

According to this theorem, each unique spectral matrix identifies a unique min-plus service with a unique spectrum, which explains why we term such matrices spectral matrices. One obvious corollary is that each min-plus service is completely identified by its spectrum. A second is that, among all worst-case services that share the same spectrum, $\Lambda = [\lambda_{ij}]_{i,j\in\mb{N}}$, there is always {\it one, and only one,} min-plus service. In fact, this service is identified by the spectral matrix $S=\Lambda$, because rewriting (\ref{E:SSig}) in matrix form yields $\Lambda(\bs\psi^S)=S=\Lambda$. Moreover, as $\Lambda(\bs\psi^S)=S$, according to (\ref{E:psiSVec}), $\bs\psi^S(\bs{q}) = \bs{q} \otimes \Lambda(\bs\psi^S)$. Comparing this to (\ref{E:psiBound+}), it is immediate that if $\Lambda(\bs\psi^S) = \Lambda(\bs\psi)$, $\bs\psi^S \geq \bs\psi$, that is, $\bs\psi^S(\bs{q}) \geq \bs\psi(\bs{q})$ for all $\bs{q} \in \mb{U}\ua b$. So, among all worst-case services that share the same spectrum, {\it the one} min-plus service identified by this spectrum is always the maximum. Notice that, as $\geq$ only defines a partial order among worst-case services, the existence of this maximum is not self-evident.

As a min-plus service is completely identified by its spectrum, so is its conditional spectrum. Denote $\lambda_{ij}(\bs\psi^S|q)$ by $\hat{s}_{ij}$. The next theorem details how $\hat{s}_{ij}$ is identified by $s_{ij}$.

\begin{theorem}\label{T:SCondSig}
	For all $i,j\in\mb{N}$,
	\begin{subnumcases}{\label{E:SCondSigCases}
		\hat{s}_{ij}=}
			\min \{s_{0j}, q+s_{1j}\}     & \textup{if $i=0$,} \label{E:SCondSigCaseA}\\
			\min \{(s_{0j}-q)^+, s_{ij}\} & \textup{if $i>0$.} \label{E:SCondSigCaseB}
	\end{subnumcases}
\end{theorem}

\subsection{A Second Definition and the Update Rule}\label{SS:MPSDef+}

In Definition~\ref{D:SM-MPS}, we imposed the four conditions in  (\ref{E:SMDefEntries}) on $S$ to mirror the four spectral properties in (\ref{E:SigPropEntries}). Following the definition, we showed that $\bs\psi^S$ is a worst-case service. But, as nowhere in this argument are (\ref{E:SMDefEntries})-(c, d) used, neither is necessary for $\bs\psi^S$ to be a worst-case service. This observation motivates the following alternative definition.

\begin{definition}\label{D:CM-MPS}
	$M=[m_{ij}]_{i,j\in\mb{N}}$ is a \textbf{cumulative matrix} if, for all $i,j\in\mb{N}$,
	\begin{equation}\label{E:CMDefEntries}
		\textup{(a)}\blank m_{ij} = 0 \textup{ if } i \geq j,
		\blank\textup{and}\blank
		\textup{(b)}\blank m_{ij} \leq m_{i, j+1}.
	\end{equation}
	For a flow with $b$ tasks left unserved prior to slot $t$, we call $\bs\psi^M$ the \textbf{min-plus service} identified by $M$ if
	\begin{equation}\label{E:psiMVec}
		\bs\psi^M(\bs{q}):= \bs{q} \otimes M \blank\blank\forall \bs{q} \in \mb{U}\ua b.
	\end{equation}
\end{definition}

As the omissions of (\ref{E:SMDefEntries})-(c, d) from this definition do not affect (\ref{E:psiS})'s derivations, for all $j\in\mb{N}$, we can replicate (\ref{E:psiS}) in terms of $M$ as
\begin{equation}\label{E:psiM+}
	\psi_j^M(\bs{q}) = \bs{q} \otimes M_{\bs\cdot j} = \min_{i\in\mb{N}} (q_i+m_{ij}) = \min_{i \leq j} (q_i+m_{ij}).
\end{equation}
But, by omitting (\ref{E:SMDefEntries})-(c, d), it would seem that we unduly enlarge the set of min-plus services. According to the next theorem, this concern is unwarranted.

\begin{theorem}\label{T:MtoS}
	For a flow with $b$ tasks left unserved prior to slot $t$, given cumulative matrix $M$, construct $S=[s_{ij}]_{i,j\in\mb{N}}$ such that
	\begin{subnumcases}{\label{E:MtoSCases}
		s_{ij} =}
			m_{0j} & \textup{if $i=0$,}\label{E:MtoSCaseA}\\
			\min \left\{(m_{0j}-b)^+, \min_{1 \leq k \leq i} m_{kj}\right\} & \textup{if $i>0$.}\label{E:MtoSCaseB}
	\end{subnumcases}
	Then $S$ is a spectral matrix, and $\bs\psi^S = \bs\psi^M$.
\end{theorem}

According to this theorem, the same min-plus service can be identified by multiple cumulative matrices. While this is redundant, it enables a simple formulation of the update rule.

\subsubsection*{\textup{[}The Update Rule\textup{]}} A class of worst-case services is {\it \textbf{update invariant}} if it is preserved by the update rule of Theorem~\ref{T:WCSUpdate}, that is, if $\bs\psi$ belongs to the class, so does $\bs{\dot\psi}$. Significantly, an update invariant class defines a closed subspace in the state space and we can confine the operation of state-based scheduling to this subspace. The next theorem shows that min-plus services are update invariant and establishes their specific update rule.

\begin{theorem}\label{T:MUpdate}
	In Theorem~\ref{T:WCSUpdate}, if $\bs\psi$ is a min-plus service that can be identified by a cumulative matrix, $M$, that is, $\bs\psi = \bs\psi^M$, then, if $q \geq d \geq p$, $\bs{\dot\psi}$ is also a min-plus service that can be identified by a cumulative matrix, $\dot{M}=[\dot{m}_{ij}]_{i,j\in\mb{N}}$, that is, $\bs{\dot\psi} = \bs\psi^{\dot{M}}$, where
	\begin{subnumcases}{\label{E:MUpdateCases}
		\dot{m}_{ij}=}
			(\min \{m_{0, j+1}, q+m_{1, j+1}\}-d)^+ & \textup{if $i=0$,}\label{E:MUpdateCaseA}\\
			m_{i+1, j+1}                            & \textup{if $i>0$.}\label{E:MUpdateCaseB}
	\end{subnumcases}
\end{theorem}

We can rewrite this update rule in terms of spectral matrices. Notice that if we treat a spectral matrix, $S$, as a general cumulative matrix and then use (\ref{E:MUpdateCases}) to find $\dot{M}$, the resulting $\dot{M}$ may not be a spectral matrix, but we can always use (\ref{E:MtoSCases}) to turn it into $\dot{S}=[\dot{s}_{ij}]_{i,j\in\mb{N}}$. This is straightforward conceptually, but does require a bit of derivation. Alternatively, once update invariance has been established, according to Theorem~\ref{T:SSig}, $\dot{s}_{ij} = \lambda_{ij}(\bs\psi^{\dot{S}})$. So, for all $i,j\in\mb{N}$, (\ref{E:SigUpdateCases}), via (\ref{E:SCondSigCases}), yields
\begin{subnumcases}{\label{E:SUpdateCases}
	\dot{s}_{ij}=}
		(\hat{s}_{0, j+1}-d)^+ =(\min \{s_{0, j+1}, q+s_{1, j+1}\}-d)^+ & \textup{if  $i=0$,}\label{E:SUpdateCaseA}\\
		\hat{s}_{i+1, j+1} =\min \{(s_{0, j+1}-q)^+, s_{i+1, j+1}\} & \textup{if $i>0$.}\label{E:SUpdateCaseB}
\end{subnumcases}

\section{Dual-Curve Services}\label{S:DVS}

In $M$'s update rule, (\ref{E:MUpdateCaseB}) is especially simple. It implies that, if $m_{ij} = m_{i+1, j+1}$ for all $i > 0$, then $\dot{m}_{ij} = m_{ij}$ for all $i > 0$. So, in this case, while the $0$th row of $M$ remains a dynamic cumulative vector, all subsequent rows are static, right-shifted versions of one another. It follows that efficiency can be further improved by simply compressing these static rows into a single, static cumulative vector. It is this observation that motivates the introduction of dual-curve services. In the remainder of this section, we first define dual-curve services and relate them to service curves. We then investigate dual-curve systems and explore their connections to EDF scheduling. We conclude the section, and the paper, by revisiting Example \ref{Ex:appetizer}.

\subsection{Definition and Relation to Service Curves}\label{SS:VPS}

\begin{definition}\label{D:DVS}
	Given a pair of cumulative vectors, \mbox{$\bs{u}, \bs{v} \in \mb{U}$}, construct a cumulative matrix, $M^{(\bs{u}, \bs{v})} = [m_{ij}^{(\bs{u}, \bs{v})}]_{i,j\in\mb{N}}$, such that
	\begin{subnumcases}{\label{E:rsCases}
		m_{ij}^{(\bs{u}, \bs{v})}:=}
			u_j             & \textup{if $i=0$,}\label{E:rsCaseA}\\
			v_{(j-i)^+}     & \textup{if $i>0$.}\label{E:rsCaseB}
	\end{subnumcases}
	The min-plus service identified by $M^{(\bs{u}, \bs{v})}$, $\bs\psi^{M^{(\bs{u}, \bs{v})}}$, is called the \textbf{dual-curve service} identified by $(\bs{u}, \bs{v})$ and denoted by~$\bs\psi^{(\bs{u}, \bs{v})}$.
\end{definition}

Most dual-curve service properties are straightforward specializations of the analogous min-plus service properties. For all $\bs{q} \in \mb{U} \ua b$ and $j\in\mb{N}$, using (\ref{E:psiM+}) and (\ref{E:rsCases}), we have
\begin{equation}\label{E:psiuv}
	\psi_j^{(\bs{u}, \bs{v})}(\bs{q}) = \min_{i \leq j} (q_i+m_{ij}^{(\bs{u}, \bs{v})}) = \min \left\{u_j, \min_{1 \leq i \leq j} (q_i+v_{j-i}) \right\}.
\end{equation}
Additionally, dual-curve services are update invariant. In Theorem~\ref{T:MUpdate}, if $M$ can be identified by $(\bs{u}, \bs{v})$, that is, $M = M^{(\bs{u}, \bs{v})}$, using (\ref{E:MUpdateCases}) and (\ref{E:rsCases}), it is easy to verify that $\dot{M}$ can be identified by $(\bs{\dot{u}}, \bs{v})$, that is, $\dot{M} = M^{(\bs{\dot{u}}, \bs{v})}$, so that $\bs{v}$, being static, remains the same, and $\bs{u}$, being dynamic, is updated to $\bs{\dot{u}}=[\dot{u}_j]_{j\in\mb{N}} \in \mb{U}$, with
\begin{equation}\label{E:uUpdateEntries}
	\text{(a)}\blank \dot{u}_j = (\hat{u}_{j+1}-d)^+, \blank\text{or}\blank
	\text{(b)}\blank \bs{\dot{u}} = \mc{R}^{-1}(\bs{\hat{u}}-d\bs\delta)^+,
\end{equation}
where
\begin{equation}\label{E:uhatEntries}
	\text{(a)}\blank \hat{u}_{j+1} := \min\{u_{j+1}, q+v_j\},
	\blank\text{or}\blank
	\text{(b)}\blank \bs{\hat{u}} := \min\{\bs{u}, q\bs\delta+\mc{R}\bs{v}\}.
\end{equation}

Consider next the spectra and conditional spectra of dual-curve services. For all $i, j \in \mb{N}$, denote $\lambda_{ij}(\bs\psi^{(\bs{u}, \bs{v})})$ by $\lambda_{ij}^{(\bs{u}, \bs{v})}$ and  $\lambda_{ij}(\bs\psi^{(\bs{u}, \bs{v})}|q)$ by $\hat\lambda_{ij}^{(\bs{u}, \bs{v})}$. So, using (\ref{E:SSig}), (\ref{E:MtoSCases}), and (\ref{E:rsCases}), we have
\begin{subnumcases}{\label{E:rsMtoSCases}
	\lambda_{ij}^{(\bs{u}, \bs{v})} = s_{ij}^{(\bs{u}, \bs{v})} =} 
		m_{0j}^{(\bs{u}, \bs{v})} = u_j                                        & \text{if $i=0$,}\label{E:rsMtoSCaseA}\\
		\min \left\{(m_{0j}^{(\bs{u}, \bs{v})}-b)^+, \min_{1 \leq k \leq i} m_{kj}^{(\bs{u}, \bs{v})}\right\} = \min \{(u_j-b)^+, v_{(j-i)^+}\} & \text{if $i>0$.}\label{E:rsMtoSCaseB}
\end{subnumcases}
Using (\ref{E:SCondSigCases}), (\ref{E:rsMtoSCases}), (\ref{E:uhatEntries}), and the fact that $q \geq b$, we also have
\begin{subnumcases}{\label{E:rsMCondSigCases}
	\hat\lambda_{ij}^{(\bs{u}, \bs{v})} = \hat{s}_{ij}^{(\bs{u}, \bs{v})} =}
		\min\{s_{0j}^{(\bs{u}, \bs{v})}, q+s_{1j}^{(\bs{u}, \bs{v})}\} =\min\{u_j, q+v_{(j-1)^+}\} =\hat{u}_j                       & \text{if $i=0$,}\label{E:rsMCondSigCaseA}\\
		\min\{(s_{0j}^{(\bs{u}, \bs{v})}-q)^+, s_{ij}^{(\bs{u}, \bs{v})}\} =\min \{(u_j-q)^+, v_{(j-i)^+}\} & \text{if $i>0$.}\label{E:rsMCondSigCaseB}
\end{subnumcases}

\subsubsection*{\textup{[}Relation to Service Curves\textup{]}} If $\bs{u} = \bs{v}$, according to (\ref{E:psiuv}), $\psi_j^{(\bs{v}, \bs{v})}(\bs{q}) = \min_{i \leq j} (q_i+v_{j-i})$. This special case recovers the well-known min-plus convolution defining service curves.\footnote{Like service curves, dual-curve services are composable, that is, multiple dual-curve services in series can be modeled by a single dual-curve service. For their composition rule, see Appendix \ref{A:composition}.} But, according to \cite{Cruz:1995}, to specify service using service curves, it is necessary that $b=0$. So, only in the case that $b=0$ is $\bs\psi^{(\bs{v}, \bs{v})}$ equivalent to the service curve specified by $\bs{v}$. As the relations $\bs{u} = \bs{v}$ and $b=0$ are not preserved by, respectively, (\ref{E:uUpdateEntries}) and (\ref{E:bUpdate}), service curves are not update invariant. Accordingly, by introducing dual-curve services, we are only extending service curves to their dynamic closure, that is, we are only allowing for cases where $\bs{u} \neq \bs{v}$ or $b>0$. But this extension is essential to state-based scheduling because we cannot update a state to something yet undefined.

Comprehensive introductions to service curves can be found in \cite{Chang:2000, Boudec:2001, Bouillard:2018}. As special cases of dual-curve services, their well-documented versatility makes clear the versatility of dual-curve services, and by extension, min-plus, and general worst-case, services.

\subsection{Dual-Curve Systems and EDF Scheduling}

A {\it \textbf{dual-curve system}} is a worst-case system in which every worst-case service is a dual-curve service. It is update invariant because its constituent dual-curve services are update invariant. For all $\omega\in\Omega$, let the dual-curve service guaranteed to flow~$\omega$ be identified by $(\bs{u}^\omega, \bs{v}^\omega)$, and denote the resulting dual-curve system by $(\bs{u}^{[\Omega]}, \bs{v}^{[\Omega]})$. Then, using (\ref{E:rsMtoSCases}), and Theorem~\ref{T:schedulability}, it is easy to verify that $(\bs{u}^{[\Omega]}, \bs{v}^{[\Omega]})$ is schedulable if and only if
\begin{equation}\label{E:uvschedulability}
	 \forall j\in\mb{N}, \blank\blank \max\left\{u_j^{\l\Omega\r}, \sum_{\omega\in\Omega} \min \left\{(u_\infty^\omega-b^\omega)^+, v_j^\omega\right\}\right\} \leq jc.
\end{equation}
In the case that $u_\infty^\omega = \infty$ for all $\omega\in\Omega$, (\ref{E:uvschedulability}) reduces to
\begin{equation}\label{E:uvschedulability+}
	\max\{\bs{u}^{\l\Omega\r}, \bs{v}^{\l\Omega\r}\} \leq \bs{c} = [c_j]_{j\in\mb{N}} := [jc]_{j\in\mb{N}}.
\end{equation}

According to (\ref{E:rsMCondSigCases}) and (\ref{E:uhatEntries})-(a), $\hat\lambda_{0,j+1}^{(\bs{u}, \bs{v})} = \hat{u}_{j+1}$ and $\hat\lambda_{1,j+1}^{(\bs{u}, \bs{v})} = \min\{(u_{j+1}-q)^+, v_j\} = (\hat{u}_{j+1}-q)^+$, which corroborate (\ref{E:CondSigPropE}). So, for all $\Gamma \subseteq \Omega$, using (\ref{E:beta}), we have
\begin{equation}\label{E:betaDVS}
	\beta(\Gamma) = \max_{j\in\mb{N}} \left(\sum_{\omega\in\Gamma}\hat\lambda_{0, j+1}^{(\bs{u}^\omega, \bs{v}^\omega)}+\sum_{\omega\in\overline\Gamma}\hat\lambda_{1, j+1}^{(\bs{u}^\omega, \bs{v}^\omega)}-c_{1, j+1}\right)	= \max_{j \in \mb{N}} \left(
	\hat{u}_{j+1}^{\l\Gamma\r} + \sum_{\omega\in\overline\Gamma} (\hat{u}_{j+1}^\omega-q^\omega)^+ - jc
	\right).
\end{equation}
Using (\ref{E:p+}), (\ref{E:rsMCondSigCaseA}), and (\ref{E:uhatEntries}), we also have
\begin{equation}\label{E:puvEntries}
	\text{(a)}\blank p_j^{(\bs{u}, \bs{v})} = \min \{\hat{u}_j, q\} = \min \{u_j, q\},
	\blank\text{or}\blank
	\text{(b)}\blank \bs{p}^{(\bs{u}, \bs{v})} = \min \{\bs{\hat{u}}, q\bs\delta\} = \min \{\bs{u}, q\bs\delta\}.
\end{equation}
These specializations enable us to use Theorems~\ref{T:polytope} and \ref{T:FP} to determine the feasible polytopes and permutohedra of dual-curve systems, and to use Definitions~\ref{D:MSS} and \ref{D:HMSS} to identify the systems' max-slack, and per-class max-slack, schedules.

Because of their efficiency, dual-curve systems approach near practical viability. So we expect state-based schedulers, when implemented, will first be implemented for such systems. Despite their efficiency, it is still impossible to implement state-based schedulers for dual-curve systems in their most general form because, well, countably infinite is still infinite. To further reduce the curves' dimensionality, we can make both $\bs{u}$ and $\bs{v}$ piecewise-linear. This piecewise-linear property is update invariant because, according to (\ref{E:uUpdateEntries}) and (\ref{E:uhatEntries}), if both $\bs{u}$ and $\bs{v}$ are piecewise-linear, so is $\bs{\dot{u}}$.

\subsubsection*{\textup{[}EDF Scheduling\textup{]}} The premise of EDF scheduling is that each task can be assigned a deadline upon arrival. Unfortunately, this premise does not hold for general worst-case systems. According to (\ref{E:tau}), to guarantee $\bs{d} \geq \bs\psi(\bs{q})$, the $h$th task queued in the buffer must be served no later than slot $t+\tau_h(\bs\psi(\bs{q}))$, $h=1,2,3,\dots,q$. The problem is that, since the scheduler is causal, it cannot foresee which $\bs{q} \in \mb{U}|q$ will be realized, and a different $\bs{q}$ may lead to a different $\tau_h(\bs\psi(\bs{q}))$. However, this is not a problem for dual-curve systems because, using (\ref{E:tau}), it is easy to verify that
\begin{equation}\label{E:tauDeadline}
	\forall	\bs{q} \in \mb{U}|q, \blank\blank \tau_h(\bs\psi^{(\bs{u}, \bs{v})}(\bs{q})) = \tau_h(\bs{p}^{(\bs{u}, \bs{v})}), \blank\blank h=1,2,3,\dots,q,
\end{equation}
by observing that, for all $\bs{q} \in \mb{U}|q$ and $j \in \mb{N}$, as $q_i \geq q$ if $1\leq i \leq j$, according to (\ref{E:psiuv}) and (\ref{E:puvEntries})-(a), $\min\{\psi_j^{(\bs{u}, \bs{v})}(\bs{q}), q\} = \min \{u_j, q\} = p_j^{(\bs{u}, \bs{v})}$. Then, by assigning the $h$th task deadline $t+\tau_h(\bs{p}^{(\bs{u}, \bs{v})})$, EDF scheduling can be applied to dual-curve systems. Moreover, when so applied, EDF scheduling results in max-slack schedules. To see this, notice that, for each $\omega\in\Omega$, using (\ref{E:MSSDefEntries}) and (\ref{E:tau}), we have
\begin{equation}\label{E:tauMSSEntries}
	\text{(a)}\blank h^\omega \leq e_\mu^\omega \text{ if } \tau_{h^\omega}(\bs{p}^\omega) < j_\mu,
	\blank\text{and}\blank
	\text{(b)}\blank h^\omega > e_\mu^\omega \text{ if } \tau_{h^\omega}(\bs{p}^\omega) > j_\mu,
	\blank\blank h^\omega = 1, 2, 3,\ldots,q^\omega.
\end{equation}
That is to say, when $e_\mu^\omega$ tasks from flow $\omega$ are served, the $h^\omega$th task is served if $\tau_{h^\omega}(\bs{p}^\omega) < j_\mu$, but is not served if $\tau_{h^\omega}(\bs{p}^\omega) > j_\mu$. As $j_\mu$ is constant across all flows, this implies that $e_\mu^{[\Omega]}$ must serve tasks in non-decreasing order of their respective $\tau_{h^\omega}(\bs{p}^\omega)$ values, which is exactly the service order dictated by EDF scheduling. 

So, for a dual-curve system, through (\ref{E:tauDeadline}) and (\ref{E:tauMSSEntries}), an EDF scheduler automatically generates max-slack schedules. According to the reasoning that follows Theorem \ref{T:MSS}, if any schedule is feasible, so is the work-conserving max-slack schedule. Thus, if any scheduler can guarantee a dual-curve system, so can the work-conserving EDF scheduler. This optimality of EDF scheduling is well-known in other contexts \cite{Georgiadis:1997, Stankovic:1998}. A corollary is that service curves, as special cases of dual-curve services, can be guaranteed by the work-conserving EDF scheduler. This is exactly what was proposed in \cite{Sariowan:1999}. Interestingly, it was realized in \cite{Stoica:2000} that, from time to time, the EDF scheduler need not be work-conserving. In our terminology, \cite{Stoica:2000} proposed to set the EDF scheduler's $\mu = \beta(\Omega)$, instead of $\min\{c, q^{\l\Omega\r}\}$, so that excess capacity, $c-\beta(\Omega)$, can be freely allocated to improve fairness.

\subsection{Revisiting Example \ref{Ex:appetizer}}

\begin{example}\label{Ex:EDF}
	In slot~$t$, the services guaranteed to flows~$1$ and $2$ in Example \ref{Ex:appetizer} can be modeled by dual-curve services $(b_0^1\mc{R}^{98}\bs\delta, \bs{0})$ and $(b_0^2\mc{R}^{99}\bs\delta, \bs{0})$, respectively, where $b_0^1=b_0^2=50c$ and $\bs{0}$ is the all-zero vector.\footnote{Consider flow~$1$. For all $\bs{q}^1 \in \mb{U} \ua b_0^1$, $\bs{q}^1 \geq \bs{u}^1 = b_0^1\mc{R}^{98}\bs\delta$, so it is immediate from (\ref{E:psiuv}) that $\bs\psi^{(\bs{u}^1, \bs{0})}(\bs{q}^1) = \bs{u}^1 = b_0^1\mc{R}^{98}\bs\delta$, that is, the $b_0^1$ tasks from flow~$1$ must be served before slot $t+99$. The same reasoning applies to flow~$2$.} Similarly, in slot $t+i$, the services can be modeled by $(b_i^1\mc{R}^{98-i}\bs\delta, \bs{0})$ and $(b_i^2\mc{R}^{99-i}\bs\delta, \bs{0})$, $i=0,1,2,\dots,98$. So, using (\ref{E:betaDVS}), the baseline function in slot $t+i$, $\beta_i$, takes the following values:\footnote{In our case, as clearly $u_\infty^{[\Omega]} \leq q^{[\Omega]}$, according to (\ref{E:uhatEntries})-(a), $\hat{u}_{j+1}^{[\Omega]} = u_{j+1}^{[\Omega]}$, so according to (\ref{E:betaDVS}), $\beta(\Gamma) = \max_{j\in\mb{N}}(u_{j+1}^{\l\Gamma\r}-jc)$.}
	\[
		\beta_i(\{1\}) = \max_{j=0, 98-i} \{[b_i^1\mc{R}^{98-i}\bs\delta]_{j+1} - jc\} = (b_i^1+ic-98c)^+,	
	\]
	\[
		\beta_i(\{2\}) = \max_{j=0, 99-i}  \{[b_i^2\mc{R}^{99-i}\bs\delta]_{j+1} - jc\} = (b_i^2+ic-99c)^+,
	\]
	\[
		\beta_i(\{1, 2\}) = \max_{j=0, 98-i, 99-i} \{[b_i^1\mc{R}^{98-i}\bs\delta]_{j+1} + [b_i^2\mc{R}^{99-i}\bs\delta]_{j+1} - jc\} = \max\{(b_i^1+ic-98c)^+, b_i^1+b_i^2+ic-99c\}.
	\]
	Denoting $\beta_i(\{1\})$, $\beta_i(\{2\})$, and $\beta_i(\{1,2\})$ by $\beta_i^1$, $\beta_i^2$, and $\beta_i^{12}$ respectively, they can be reduced to
	\begin{equation}\label{E:beta12}
		\textup{(a)}\blank \beta_i^1 = (2c-b_i^2)^+,
		\blank
		\textup{(b)}\blank \beta_i^2 = (c-b_i^1)^+,
		\blank\textup{and}\blank
		\textup{(c)}\blank \beta_i^{12} = c,
		\blank\blank i=0,1,2,\ldots,98,
	\end{equation}
	by observing that the following relations hold:\footnote{If (\ref{E:b12}) holds, $b_i^1+ic-98c = 2c-b_i^2 \leq c$, $b_i^2+ic-99c = c-b_i^1$, and $b_i^1+b_i^2+ic-99c = c$, so (\ref{E:beta12}) follows.}
	\begin{equation}\label{E:b12}
		\textup{(a)}\blank b_i^1 \leq (99 - i)c,
		\blank
		\textup{(b)}\blank b_i^2 \geq c,
		\blank\textup{and}\blank
		\textup{(c)}\blank b_i^1+b_i^2 = (100 - i)c, 
		\blank\blank i=0,1,2,\ldots,99.\footnote{Notice that the schedulability condition, (\ref{E:uvschedulability}), only requires that $b_i^1 \leq (99 - i)c$ and $b_i^1+b_i^2 \leq (100 - i)c$, so (\ref{E:b12}) essentially says that the states of our system are confined to the {\it Pareto} frontier in Fig. \ref{F:StatePath}.}
	\end{equation}
	To see why, suppose that they hold for $i$, and consider a feasible schedule in slot $t+i$, $[d^1(i), d^2(i)]$. According to (\ref{E:capacity}), (\ref{E:beta}), and (\ref{E:beta12})-(c), $c \geq d^1(i)+d^2(i) \geq \beta_i^{12} = c$, so $d^1(i)+d^2(i) = c$, that is, the server must be work-conserving. Then, as $b_{i+1}^1 = b_i^1-d^1(i)$ and $b_{i+1}^2 = b_i^2-d^2(i)$, according to (\ref{E:b12})-(c), $b_{i+1}^1+b_{i+2}^2 = (99 - i)c$. Moreover, as $b_{i+1}^2 = b_i^2-d^2(i) = b_i^2 - c + d^1(i)$, according to (\ref{E:beta}) and (\ref{E:beta12})-(a), $b_{i+1}^2 \geq b_i^2 - c + \beta_i^1 \geq c$, implying that $b_{i+1}^1 \leq (98 - i)c$. So, if (\ref{E:b12}) holds for $i$, it holds for $i+1$. Clearly (\ref{E:b12}) holds for $i=0$, so by induction, it holds for $i=0,1,2,\ldots,99$.
	
	According to (\ref{E:b12}), $b_{99}^1=0$ and $b_{99}^2=c$, so in slot $t+99$, a state-based scheduler has no choice but to serve $c$~tasks from flow~$2$. During interval $[t, t+99)$, the scheduler's flexibility is completely characterized by (\ref{E:beta12}). Specifically, in slot $t+i$, $d^1(i)+d^2(i) = c$, and according to (\ref{E:beta}), $d^1(i) \geq \beta_i^1$ and $d^2(i) \geq \beta_i^2$, so the feasible polytope is a line segment.\footnote{In general, the feasible polytope for a two-flow system is a polygon, such as the hexagon, $ABCDEF$, in Fig. \ref{F:FeasibleRegion}. But, in our case, since the server must be work-conserving, the polygon reduces to a line segment.} One endpoint, $[\beta_i^1, c-\beta_i^1]$, prioritizes flow~$2$, while the other, $[c-\beta_i^2, \beta_i^2]$, prioritizes flow~$1$. To fine-tune the balance, any point along the line segment can be selected. For instance, to ensure fairness, the midpoint, $[c+\beta_i^1-\beta_i^2, c-\beta_i^1+\beta_i^2]/2$, can be selected. If the midpoint is consistently selected, using (\ref{E:beta12}), it is easy to verify, by induction, that for $i = 0, 1, 2, \ldots, 96$, $b_i^1 = b_i^2 = (100-i)c/2$, $\beta_i^1 = \beta_i^2 = 0$, and $d^1(i) = d^2(i) = c/2$. This implies that $b_{97}^1 = b_{97}^2 = 3c/2$, $\beta_{97}^1 = c/2$, $\beta_{97}^2 = 0$, $d^1(97) = 3c/4$, and $d^2(97) = c/4$, which in turn implies that $b_{98}^1 = 3c/4$, $b_{98}^2 = 5c/4$, $\beta_{98}^1 = 3c/4$, $\beta_{98}^2 = c/4$, $d^1(98) = 3c/4$, and $d^2(98) = c/4$. That is to say, the server's capacity is split half-and-half during interval $[t, t+97)$, and at a $3:1$ ratio in slots $t+97$ and $t+98$. In contrast, an EDF scheduler has no choice but to prioritize flow~$1$ over flow~$2$. Notice that \cite{Stoica:2000}'s proposal is of no help in this case because the server must be work-conserving.
\end{example}

This example has simply elaborated on the smart scheduler's flexibility as outlined in Example~\ref{Ex:appetizer}. So, at least in this case, all the formalism may seem to be overkill. However, in more general cases, for instance, in the case of general piecewise-linear dual-curve systems, characterizing and exploiting the scheduler's flexibility cannot be so easy, which is why state-based scheduling is introduced.

\newpage
\bibliographystyle{ACM-Reference-Format}
\bibliography{ReferencesACM}

\newpage
\numberwithin{equation}{section}
\numberwithin{figure}{section}
\appendix

\section{Proofs for Section \ref{S:WCS}}\label{A:WCSProof}

\begin{proof}[Proof of Theorem~\ref{T:WCSUpdate}]
	The necessity of (\ref{E:WCSUpdateB}) follows directly from our argument leading to this theorem. So we need only show its sufficiency. As $d \geq p$, (\ref{E:WCSImmediate}) implies that $(\bs\psi(\bs{q})-d\bs\delta)^+ \in \mb{U}|0$ for all $\bs{q} \in \mb{U}|q$, justifying the use of $\mc{R}^{-1}$ in (\ref{E:WCSUpdateB}). Now, if (\ref{E:WCSUpdateB}) holds, for all $\bs{q} \in \mb{U}|q$, (\ref{E:vecdUpdate}) implies that
	\[
	\bs{d} = \mc{R}\bs{\dot{d}}+d\bs\delta \geq (\bs\psi(\bs{q})-d\bs\delta)^+ +d\bs\delta \geq \bs\psi(\bs{q}).
	\]
	So $\bs\psi$ is guaranteed.	It remains to show that $\bs{\dot\psi}$ is a worst-case service. For all $\bs{\dot{q}} \in \mb{U}\ua\dot{b}$ and $j>0$, using (\ref{E:WCSUpdateB}), (\ref{E:WCSDef}), and (\ref{E:qUpdateCaseB}), we have
	\[
	\dot\psi_j(\bs{\dot{q}}) = (\psi_{j+1}(\bs{q})-d)^+ \leq (q_{j+1}-d)^+ = \dot{q}_j^+ = \dot{q}_j,
	\]
	where the final equality holds because, according to (\ref{E:bUpdate}), \mbox{$d \leq q$} guarantees that $\dot{b} \geq 0$, and thus $\dot{q}_j \geq 0$. It follows that $\bs{\dot\psi}(\bs{\dot{q}}) \leq \bs{\dot{q}}$. So, according to Definition~\ref{D:WCS}, $\bs{\dot\psi}$ is indeed a worst-case service.
\end{proof}

\begin{proof}[Proof of Theorem~\ref{T:SigProp}]
	First, if $i \geq j$, since $\psi_j(\bs{q}) \leq \psi_i(\bs{q}) \leq q_i$, according to (\ref{E:signature}), $\lambda_{ij} = 0$. Second, since $\psi_j(\bs{q}) \leq \psi_{j+1}(\bs{q})$ and $q_i \leq q_{i+1}$, according to (\ref{E:signature}), $\lambda_{ij} \leq \lambda_{i, j+1}$ and $\lambda_{ij} \geq \lambda_{i+1, j}$. Finally, if $i>0$, since $q_i \geq b$ for all $\bs{q}\in\mb{U}\ua b$, using (\ref{E:signature}), we have
	\[
	\lambda_{ij} \leq \max_{\bs{q} \in \mb{U} \ua b} (\psi_j(\bs{q})-b)^+ = \left(\max_{\bs{q} \in \mb{U} \ua b} \psi_j(\bs{q})-b\right)^+ = (\lambda_{0j}-b)^+.
	\]
\end{proof}

\begin{proof}[Proof of Theorem~\ref{T:SigUpdate}]
	This theorem follows directly from our argument leading to it.
\end{proof}

\begin{proof}[Proof of Theorem~\ref{T:CondSigProp}]
	First, if $i \geq j$, since $\psi_j(\bs{q}) \leq \psi_i(\bs{q}) \leq q_i$, according to (\ref{E:CondSig}), $\hat\lambda_{ij} = 0$. Second, since $\psi_j(\bs{q}) \leq \psi_{j+1}(\bs{q})$ and $q_i \leq q_{i+1}$, according to (\ref{E:CondSig}), $\hat\lambda_{ij} \leq \hat\lambda_{i, j+1}$ and $\hat\lambda_{ij} \geq \hat\lambda_{i+1, j}$. Finally, if $i>0$, since $q_i \geq q$ for all $\bs{q}\in\mb{U}|q$, using (\ref{E:CondSig}), we have
	\[
	\hat\lambda_{ij} \leq \max_{\bs{q} \in \mb{U}|q} (\psi_j(\bs{q})-q)^+ = \left(\max_{\bs{q} \in \mb{U}|q} \psi_j(\bs{q})-q\right)^+ = (\hat\lambda_{0j}-q)^+.
	\]
\end{proof}

\section{A Proof for Section \ref{S:schedulability}}\label{A:schedulabilityProof}

\begin{proof}[Proof of Theorem~\ref{T:schedulability}]
	From our argument leading to this theorem, if (\ref{E:schedulability}) is not satisfied, $\bs\psi^{[\Omega]}$ cannot be guaranteed. So, according to Definition~\ref{D:schedulability}, we need only show that, if (\ref{E:schedulability}) is satisfied, that is, if $\bs\psi^{[\Omega]}$ is schedulable, given any $a^{[\Omega]}$, at least one feasible schedule exists. This will be established by Theorem~\ref{T:FP}. Given that (\ref{E:schedulability}) is the schedulability condition, (\ref{E:schedulability+}) follows directly from Definition~\ref{D:schedulability}. However, comparing the feasible schedule requirements to those of Definition~\ref{D:schedulability}, it is immediate that the requirement that $d^{[\Omega]} \geq p^{[\Omega]}$ has been dropped. This is because $d^{[\Omega]} \geq p^{[\Omega]}$ is implied by (\ref{E:schedulability+}). To see this, notice that, according to (\ref{E:schedulability+}), $\dot\lambda_{00}^{\l\Omega\r} \leq \dot{c}_{00} = 0$. So, according to (\ref{E:SigUpdateCaseA}) and (\ref{E:p-lambda}),
	\[
	\dot\lambda_{00}^{\l\Omega\r} = \sum_{\omega\in\Omega} (\hat\lambda_{01}^\omega-d^\omega)^+ = \sum_{\omega\in\Omega} (p^\omega-d^\omega)^+ \leq 0,
	\]
	implying that $d^{[\Omega]} \geq p^{[\Omega]}$.
\end{proof}

\section{Proofs for Section \ref{S:feasible}}\label{A:feasibleProof}

\begin{proof}[Proof of Theorem~\ref{T:polytope}]
	The necessity of this condition follows directly from (\ref{E:beta})'s derivation. To establish its sufficiency, according to Theorem~\ref{T:schedulability}, we need only show that this condition ensures that $\dot\lambda_{ij}^{\l\Omega\r} \leq c_{i+1, j+1}$ for all $i,j \in \mb{N}$. In the case that $i>0$, this holds by default because, if $\bs\psi^{[\Omega]}$ is schedulable, according to (\ref{E:SigUpdateCaseB}), (\ref{E:CondSigPropF}), and (\ref{E:schedulability}), $\dot\lambda_{ij}^{\l\Omega\r} = \hat\lambda_{i+1, j+1}^{\l\Omega\r} \leq  \lambda_{i+1, j+1}^{\l\Omega\r} \leq c_{i+1, j+1}$. In the case that $i=0$, on the one hand, using (\ref{E:SigUpdateCaseA}), and applying the fact that $\max\{{x, y\}}+z = \max\{x+z, y+z\}$ repeatedly, we have
	\[
	\dot\lambda_{0j}^{\l\Omega\r} = \sum_{\omega\in\Omega} \max\{\hat\lambda_{0, j+1}^\omega-d^\omega, 0\} = \max_{\Gamma\subseteq\Omega} (\hat\lambda_{0, j+1}^{\l\Gamma\r}-d^{\l\Gamma\r}).
	\]
	On the other hand, if (\ref{E:beta}) holds, for all $\Gamma\subseteq\Omega$, as $\hat\lambda_{ij}$ is non-negative by definition, we have
	\[
	d^{\l\Gamma\r} \geq \beta(\Gamma) \geq \hat\lambda_{0, j+1}^{\l\Gamma\r} + \hat\lambda_{1, j+1}^{\l\overline\Gamma\r}-c_{1, j+1} \geq \hat\lambda_{0, j+1}^{\l\Gamma\r}-c_{1, j+1}.
	\]
	It follows that
	\[
	\dot\lambda_{0j}^{\l\Omega\r} = \max_{\Gamma\subseteq\Omega} (\hat\lambda_{0, j+1}^{\l\Gamma\r}-d^{\l\Gamma\r}) \leq c_{1, j+1}.
	\]
	So $\dot\lambda_{ij}^{\l\Omega\r} \leq c_{i+1, j+1}$ for all $i,j \in \mb{N}$.
\end{proof}

\begin{proof}[Proof of Theorem~\ref{T:betaProp}]
	To prove this theorem, alternative formulations of $\beta$ are useful. To derive them, notice that, for all $j\in\mb{N}$, using (\ref{E:p+}) and (\ref{E:CondSigPropE}), we have
	\begin{equation}\label{E:p}
		p_j = \min\{\hat\lambda_{0j}, q\} = \hat\lambda_{0j} - (\hat\lambda_{0j}-q)^+ = \hat\lambda_{0j} - \hat\lambda_{1j}.
	\end{equation}
	It is then immediate from (\ref{E:beta}) that, for all $\Gamma\subseteq\Omega$,
	\begin{equation}\label{E:beta+}
		\beta(\Gamma) = \max_{j \in \mb{N}} (p_{j+1}^{\l\Gamma\r} + \hat\lambda_{1, j+1}^{\l\Omega\r} - c_{1, j+1}),
	\end{equation}
	so
	\begin{equation}\label{E:beta++}
		\beta(\Gamma) = p_{j_\beta^\Gamma+1}^{\l\Gamma\r} + \hat\lambda_{1, j_\beta^\Gamma+1}^{\l\Omega\r} -c_{1, j_\beta^\Gamma+1},
	\end{equation}
	where
	\begin{equation}\label{E:jbeta}
		j_\beta^\Gamma := \argmax_{j \in \mb{N}} (p_{j+1}^{\l\Gamma\r} + \hat\lambda_{1, j+1}^{\l\Omega\r} - c_{1, j+1}),
	\end{equation}
	with the understanding that, whenever comparisons result in ties, $j_\beta^\Gamma$ should be minimized.
	
	To show that $\beta$ is supermodular, according to Definition \ref{D:supermodular},  we need to show that $\beta:2^\Omega \rightarrow \mb{N}$ satisfies that
	\begin{equation}\label{E:betaPropEntries}
		\textup{(a)}\blank \beta(\phi) = 0,
		\blank\textup{and}\blank
		\textup{(b)}\blank  \forall \Gamma,\Gamma'\subseteq\Omega, ~\beta(\Gamma)+\beta(\Gamma') \leq \beta(\Gamma+\Gamma')+\beta(\Gamma\Gamma').
	\end{equation}
	By definition, $\beta$ can only take integer values, so we need only show it to be non-negative to restrict its range to $\mb{N}$. For all $\Gamma\subseteq\Omega$, using (\ref{E:beta}), and the facts that $\hat\lambda_{ij} \geq 0$ and $c_{11}=0$, we have
	\[
	\beta(\Gamma) \geq (\hat\lambda_{0, j+1}^{\l\Gamma\r} + \hat\lambda_{1, j+1}^{\l\overline\Gamma\r}-c_{1, j+1})|_{j=0} = \hat\lambda_{01}^{\l\Gamma\r} + \hat\lambda_{11}^{\l\overline\Gamma\r} -c_{11} \geq 0.
	\]
	So $\beta$ is indeed non-negative. It follows that $\beta(\phi) \geq 0$. But, if $\bs\psi^{[\Omega]}$ is schedulable, using (\ref{E:beta++}), (\ref{E:CondSigPropF}), and (\ref{E:schedulability}), we also have
	\[
	\beta(\phi) = \hat\lambda_{1, j_\beta^\phi+1}^{\l\Omega\r} - c_{1, j_\beta^\phi+1} \leq \lambda_{1, j_\beta^\phi+1}^{\l\Omega\r} - c_{1, j_\beta^\phi+1} \leq 0.
	\]
	So $\beta(\phi) = 0$. It remains to show that (\ref{E:betaPropEntries})-(b) holds. For all $\Gamma,\Gamma'\subseteq\Omega$, assume, without loss of generality, that $j_\beta^\Gamma \geq j_\beta^{\Gamma'}$. Then,
	\[
	p_{j_\beta^\Gamma+1}^{\l\Gamma\r} + p_{j_\beta^{\Gamma'}+1}^{\l\Gamma'\r} = p_{j_\beta^\Gamma+1}^{\l\Gamma\r} + p_{j_\beta^{\Gamma'}+1}^{\l\Gamma'\setminus\Gamma\r} + p_{j_\beta^{\Gamma'}+1}^{\l\Gamma\Gamma'\r}
	\leq p_{j_\beta^\Gamma+1}^{\l\Gamma\r} + p_{j_\beta^\Gamma+1}^{\l\Gamma'\setminus\Gamma\r} + p_{j_\beta^{\Gamma'}+1}^{\l\Gamma\Gamma'\r}
	= p_{j_\beta^\Gamma+1}^{\l\Gamma+\Gamma'\r} + p_{j_\beta^{\Gamma'}+1}^{\l\Gamma\Gamma'\r}.
	\]
	So, according to (\ref{E:beta++}) and (\ref{E:beta+}),
	\[
	\begin{IEEEeqnarraybox}[][c]{rCl}
		\beta(\Gamma)+\beta(\Gamma')
		& =    & p_{j_\beta^\Gamma+1}^{\l\Gamma\r} + \hat\lambda_{1, j_\beta^\Gamma+1}^{\l\Omega\r} - c_{1, j_\beta^\Gamma+1} + p_{j_\beta^{\Gamma'}+1}^{\l\Gamma'\r} + \hat\lambda_{1, j_\beta^{\Gamma'}+1}^{\l\Omega\r} - c_{1, j_\beta^{\Gamma'}+1}\\
		& \leq & p_{j_\beta^\Gamma+1}^{\l\Gamma+\Gamma'\r} + \hat\lambda_{1, j_\beta^\Gamma+1}^{\l\Omega\r} - c_{1, j_\beta^\Gamma+1} + p_{j_\beta^{\Gamma'}+1}^{\l\Gamma\Gamma'\r} + \hat\lambda_{1, j_\beta^{\Gamma'}+1}^{\l\Omega\r} - c_{1, j_\beta^{\Gamma'}+1}\\
		& \leq & \beta(\Gamma+\Gamma')+\beta(\Gamma\Gamma').
	\end{IEEEeqnarraybox}
	\]
\end{proof}

\begin{proof}[Proof of Theorem~\ref{T:FP}]
	To prove this theorem, notice that, for all $\Gamma,\Gamma'\subseteq\Omega$, using (\ref{E:beta++}) and (\ref{E:beta+}), we have
	\[
	\beta(\Gamma)-\beta(\Gamma\Gamma') \leq p_{j_\beta^\Gamma+1}^{\l\Gamma\r} + \hat\lambda_{1, j_\beta^\Gamma+1}^{\l\Omega\r} - c_{1, j_\beta^\Gamma+1} - (p_{j_\beta^\Gamma+1}^{\l\Gamma\Gamma'\r} + \hat\lambda_{1, j_\beta^\Gamma+1}^{\l\Omega\r} - c_{1, j_\beta^\Gamma+1}) =  p_{j_\beta^\Gamma+1}^{\l\Gamma\r}-p_{j_\beta^\Gamma+1}^{\l\Gamma\Gamma'\r} = p_{j_\beta^\Gamma+1}^{\l\Gamma\setminus\Gamma'\r}.
	\]
	So, according to (\ref{E:p+}),
	\begin{equation}\label{E:betaProp+}
		\beta(\Gamma)-\beta(\Gamma\Gamma') \leq q^{\l\Gamma\setminus\Gamma'\r},
	\end{equation}
	implying that $\beta(\Gamma)$ and $\beta(\Gamma\Gamma')$, as lower bounds of $d^{\l\Gamma\r}$ and $d^{\l\Gamma\Gamma'\r}$, can be {\it jointly} tight without contradicting the causality constraint, (\ref{E:causalityB}).
		
	To show that $\beta_\mu$ is supermodular, according to Definition \ref{D:supermodular},  we need to show that $\beta_\mu:2^\Omega \rightarrow \mb{N}$ satisfies that
	\begin{equation}\label{E:betaMuPropEntries}
		\textup{(a)}\blank \beta_\mu(\phi) = 0,
		\blank\textup{and}\blank
		\textup{(b)}\blank  \forall \Gamma,\Gamma'\subseteq\Omega, ~\beta_\mu(\Gamma)+\beta_\mu(\Gamma') \leq \beta_\mu(\Gamma+\Gamma')+\beta_\mu(\Gamma\Gamma').
	\end{equation}
	If $\bs\psi^{[\Omega]}$ is schedulable, according to Theorem~\ref{T:betaProp}, $\beta$ is supermodular. It is then immediate from (\ref{E:betaMu}) that, not only can $\beta_\mu$ only take integer values but also, for all $\Gamma\subseteq\Omega$, $\beta_\mu(\Gamma) \geq \beta(\Gamma) \geq 0$. So the range of $\beta_\mu$ is restricted to $\mb{N}$. Additionally, since $\beta(\phi) = 0$, and since, according to (\ref{E:muFeasible}), $\mu$'s feasibility implies that $\mu \leq q^{\l\Omega\r}$, according to (\ref{E:betaMu}),
	\[
	\beta_\mu(\phi) = \max\{\beta(\phi), \mu-q^{\l\Omega\r}\} = 0.
	\]
	It remains to show that (\ref{E:betaMuPropEntries})-(b) holds. We need only consider four cases:
	\begin{itemize}
		\item[C1] if $\beta_\mu(\Gamma) = \beta(\Gamma)$ and $\beta_\mu(\Gamma') = \beta(\Gamma')$, (\ref{E:betaMuPropEntries})-(b) follows directly from (\ref{E:betaPropEntries})-(b);
		\item[C2] if $\beta_\mu(\Gamma) = \mu-q^{\l\overline\Gamma\r}$ and $\beta_\mu(\Gamma') = \mu-q^{\l\overline{\Gamma'}\r}$, according to (\ref{E:betaMu}),
		\[
		\beta_\mu(\Gamma)+\beta_\mu(\Gamma') = \mu-q^{\l\overline\Gamma\r}+\mu-q^{\l\overline{\Gamma'}\r}= \mu-q^{\l\overline{\Gamma+\Gamma'}\r}+\mu-q^{\l\overline{\Gamma\Gamma'}\r} \leq  \beta_\mu(\Gamma+\Gamma')+\beta_\mu(\Gamma\Gamma');
		\]
		\item[C3] if $\beta_\mu(\Gamma) = \beta(\Gamma)$ and $\beta_\mu(\Gamma') = \mu-q^{\l\overline{\Gamma'}\r}$, according to (\ref{E:betaProp+}) and (\ref{E:betaMu}),
		\[
		\blank\blank\blank\blank\beta_\mu(\Gamma)+\beta_\mu(\Gamma') \leq  \beta(\Gamma\Gamma')+q^{\l\Gamma\setminus\Gamma'\r} + \mu-q^{\l\overline{\Gamma'}\r} =  \mu-q^{\l\overline{\Gamma+\Gamma'}\r}+\beta(\Gamma\Gamma') \leq \beta_\mu(\Gamma+\Gamma')+\beta_\mu(\Gamma\Gamma');
		\]
		and finally,
		\item[C4] if $\beta_\mu(\Gamma) = \mu-q^{\l\overline\Gamma\r}$ and $\beta_\mu(\Gamma') = \beta(\Gamma')$, aside from interchanging $\Gamma$ and $\Gamma'$, C4 is C3.
	\end{itemize}
	So, in all cases, (\ref{E:betaMuPropEntries})-(b) must hold.
	
	Next we show that $\mb{F}_\mu = \mb{P}(\beta_\mu)$. From our argument leading to this theorem, $\mb{F}_\mu \subseteq \mb{P}(\beta_\mu)$, so we need only show that $\mb{P}(\beta_\mu) \subseteq \mb{F}_\mu$. According to Theorem~\ref{T:polytope}, we then need only show that (\ref{E:betaMu}) implies (\ref{E:causalityB}), (\ref{E:capacity}), and (\ref{E:beta}). First, for all $\omega\in\Omega$, using (\ref{E:betaMu}), we have
	\[
	d_\mu^\omega = d_\mu^{\l\Omega\r} - d_\mu^{\l\overline{\{\omega\}}\r} \leq \mu-\beta_\mu(\overline{\{\omega\}}) = \mu-\max\{\beta(\overline{\{\omega\}}), \mu-q^\omega\} \leq q^\omega,
	\]
	so (\ref{E:causalityB}) holds. Second, according to (\ref{E:muFeasible}), $\mu$'s feasibility implies that  $d_\mu^{\l\Omega\r} = \mu \leq c$, so (\ref{E:capacity}) holds. Finally, according to (\ref{E:betaMu}), $d_\mu^{\l\Gamma\r} \geq \beta_\mu(\Gamma) \geq \beta(\Gamma)$ for all $\Gamma\subseteq\Omega$, so (\ref{E:beta}) holds.
\end{proof}

\section{Proofs for Section \ref{S:MSS}}\label{A:MSSProof}

\begin{proof}[Proof of Theorem~\ref{L:pFeasible}]
	Clearly $p_j^{[\Omega]}$ is valid if $p_j^{\l\Omega\r} \leq c$ because, according to (\ref{E:p+}), $p_j^{[\Omega]} \leq q^{[\Omega]}$. So, according to Theorem~\ref{T:polytope}, we need only show that $p_j^{\l\Omega\r} \geq \beta(\Omega)$ implies that $p_j^{\l\Gamma\r} \geq \beta(\Gamma)$ for all $\Gamma\subseteq\Omega$. To see this, on the one hand, if $\bs\psi^{[\Omega]}$ is schedulable, according to (\ref{E:CondSigPropF}) and (\ref{E:schedulability}), $\hat\lambda_{1, j_\beta^\Gamma+1}^{\l\Omega\r} \leq \lambda_{1, j_\beta^\Gamma+1}^{\l\Omega\r} \leq c_{1, j_\beta^\Gamma+1}$. So, if $j > j_\beta^\Gamma$, (\ref{E:beta++}) implies that $p_j^{\l\Gamma\r} \geq p_{j_\beta^\Gamma+1}^{\l\Gamma\r} \geq \beta(\Gamma)$. On the other hand, since $p_j^{\l\Omega\r} \geq \beta(\Omega)$, if $j \leq j_\beta^\Gamma$, using (\ref{E:beta+}) and (\ref{E:beta++}), we have
	\[
	p_j^{\l\Gamma\r} = p_j^{\l\Omega\r} - p_{j}^{\l\overline\Gamma\r} \geq \beta(\Omega) - p_{j_\beta^\Gamma+1}^{\l\overline\Gamma\r} \geq p_{j_\beta^\Gamma+1}^{\l\Omega\r} + \hat\lambda_{1, j_\beta^\Gamma+1}^{\l\Omega\r} -c_{1, j_\beta^\Gamma+1} - p_{j_\beta^\Gamma+1}^{\l\overline\Gamma\r} = p_{j_\beta^\Gamma+1}^{\l\Gamma\r} + \hat\lambda_{1, j_\beta^\Gamma+1}^{\l\Omega\r} -c_{1, j_\beta^\Gamma+1} = \beta(\Gamma).
	\]
\end{proof}

\begin{proof}[Proof of Theorem~\ref{T:MSS}]
	If $\mb{F}_\mu$ is non-empty, given any $d_\mu^{[\Omega]}\in\mb{F}_\mu$, as $d_\mu^{[\Omega]}$ is feasible, $d_\mu^{[\Omega]} \leq q^{[\Omega]}$, so $\mu = d_\mu^{\l\Omega\r} \leq q^{\l\Omega\r}$. Then, comparing (\ref{E:jMu+Entries}) to (\ref{E:MSSDefEntries}), it is immediate that, since $\mu$, $\bs{p}^{[\Omega]}$, and $q^{[\Omega]}$ are all integral, there exists at least one way to split $\mu$ integrally into $e_\mu^{[\Omega]}$. So $e_\mu^{[\Omega]}$ exists. By default, (\ref{E:MSSDefEntries}) requires that $e_\mu^{[\Omega]} \leq q^{[\Omega]}$. Given any \mbox{$d_\mu^{[\Omega]}\in\mb{F}_\mu$}, as $d_\mu^{[\Omega]}$ is feasible, $e_\mu^{\l\Omega\r} = \mu = d_\mu^{\l\Omega\r} \leq c$. So $e_\mu^{[\Omega]}$ is valid. Also, as $d_\mu^{[\Omega]}$ is feasible, according to (\ref{E:schedulability+}), $\dot\lambda_{ij}^{\l\Omega\r}(d_\mu^{[\Omega]}) \leq \dot{c}_{ij}$ for all $i,j \in \mb{N}$. So Theorem~\ref{L:MSS} guarantees that $\dot\lambda_{ij}^{\l\Omega\r}(e_\mu^{[\Omega]}) \leq \dot\lambda_{ij}^{\l\Omega\r}(d_\mu^{[\Omega]}) \leq \dot{c}_{ij}$ for all $i,j \in \mb{N}$. Then, according Theorem~\ref{T:schedulability}, $e_\mu^{[\Omega]}$ is feasible, and $e_\mu^{[\Omega]} \in \mb{F}_\mu$. Notice that the use of Theorem~\ref{L:MSS} here is not circular because, as we will see next, the proof of Theorem~\ref{L:MSS} is independent of this theorem.
\end{proof}

\begin{proof}[Proof of Theorem~\ref{L:MSS}]
	If $e_\mu^{[\Omega]}$ exists, (\ref{E:MSSDefEntries}) requires that $e_\mu^{[\Omega]} \leq q^{[\Omega]}$, so at least one $d_\mu^{[\Omega]} \leq q^{[\Omega]}$ exists. If $i>0$, it is immediate from (\ref{E:SigUpdateCaseB}) that (\ref{E:MSS}) holds by default because $\dot\lambda_{ij}^{\l\Omega\r}$ is constant with respect to $d_\mu^{[\Omega]}$. If $i = 0$, an alternative formulation of $\dot\lambda_{0j}$ is useful. Given any $d \leq q$, according to (\ref{E:lambdaBound}), $\dot\lambda_{0j} \geq \hat\lambda_{1, j+1}$. So, using (\ref{E:SigUpdateCaseA}) and (\ref{E:p}), we have
	\[
	\dot\lambda_{0j} = \hat\lambda_{1, j+1} + (\dot\lambda_{0j} - \hat\lambda_{1, j+1})^+
	= \hat\lambda_{1, j+1} + (\hat\lambda_{0, j+1} - d - \hat\lambda_{1, j+1})^+
	= \hat\lambda_{1, j+1} + (p_{j+1} - d)^+.
	\]
	Then, given any $d_\mu^{[\Omega]} \leq q^{[\Omega]}$,
	\[
	\dot\lambda_{0j}^{\l\Omega\r}(d_\mu^{[\Omega]}) = \hat\lambda_{1,j+1}^{\l\Omega\r} + \sum_{\omega\in\Omega} (p_{j+1}^\omega-d_\mu^\omega)^+ \geq \hat\lambda_{1,j+1}^{\l\Omega\r} + (p_{j+1}^{\l\Omega\r}-\mu)^+.
	\]
	This bound can be achieved for all $j\in\mb{N}$ simultaneously when $d_\mu^{[\Omega]} = e_\mu^{[\Omega]}$ because, if $j+1 \leq j_\mu$, according to (\ref{E:MSSDefEntries}) and (\ref{E:jMu+Entries}), $e_\mu^{[\Omega]} \geq p_{j+1}^{[\Omega]}$ and $\mu \geq p_{j+1}^{\l\Omega\r}$, implying that both sides equal $\hat\lambda_{1,j+1}^{\l\Omega\r}$, while, if $j+1 > j_\mu$, $e_\mu^{[\Omega]} \leq p_{j+1}^{[\Omega]}$ and $\mu < p_{j+1}^{\l\Omega\r}$, implying that both sides equal $\hat\lambda_{1,j+1}^{\l\Omega\r} + p_{j+1}^{\l\Omega\r}-\mu$. So (\ref{E:MSS}) holds.
\end{proof}

\begin{proof}[Proof of Theorem~\ref{T:HMSS}]
	If $\mb{F}(\nu_\mu^{[\mc{P}]})$ is non-empty, given any $d_\mu^{[\Omega]}\in\mb{F}(\nu_\mu^{[\mc{P}]})$, as $d_\mu^{[\Omega]}$ is feasible, $d_\mu^{[\Omega]} \leq q^{[\Omega]}$. So, for each $\Gamma\in\mc{P}$, $\nu_\mu^\Gamma = d_\mu^{\l\Gamma\r} \leq q^{\l\Gamma\r}$. Then, using (\ref{E:jT}) and (\ref{E:tau}), it is easy to verify that
	\begin{equation}\label{E:jT+Entries}
		\textup{(a)}\blank \nu_\mu^\Gamma \in [p_{j_\mu^\Gamma}^{\l\Gamma\r}, p_{j_\mu^\Gamma+1}^{\l\Gamma\r}) \textup{ if } j_\mu^\Gamma < \infty,
		\blank\textup{and}\blank
		\textup{(b)}\blank \nu_\mu^\Gamma \in [p_\infty^{\l\Gamma\r}, q^{\l\Gamma\r}] \text{ if } j_\mu^\Gamma = \infty.
	\end{equation}
	Comparing this to (\ref{E:HMSSDefEntries}), since $\nu_\mu^\Gamma$, $\bs{p}^{[\Gamma]}$, and $q^{[\Gamma]}$ are all integral, there exists at least one way to split $\nu_\mu^\Gamma$ integrally into $e^{[\Gamma]}(\nu_\mu^{[\mc{P}]})$. So $e^{[\Gamma]}(\nu_\mu^{[\mc{P}]})$ exists, and by concatenation, $e^{[\Omega]}(\nu_\mu^{[\mc{P}]})$ exists.
	
	By default, (\ref{E:HMSSDefEntries}) requires that $e^{[\Omega]}(\nu_\mu^{[\mc{P}]}) \leq q^{[\Omega]}$. Given any $d_\mu^{[\Omega]}\in\mb{F}(\nu_\mu^{[\mc{P}]})$, as $d_\mu^{[\Omega]}$ is feasible, $e^{\l\Omega\r}(\nu_\mu^{[\mc{P}]}) = \nu_\mu^{\l\mc{P}\r} = \mu = d_\mu^{\l\Omega\r} \leq c$. So $e^{[\Omega]}(\nu_\mu^{[\mc{P}]})$ is valid. Also, as $d_\mu^{[\Omega]}$ is feasible, according to (\ref{E:schedulability+}), $\dot\lambda_{ij}^{\l\Omega\r}(d_\mu^{[\Omega]}) \leq \dot{c}_{ij}$ for all $i,j \in \mb{N}$. Additionally, replicating, almost verbatim, Theorem~\ref{L:MSS}'s proof,  for each $\Gamma\in\mc{P}$, we have
	\begin{equation}\label{E:HMSS}
		\dot\lambda_{ij}^{\l\Gamma\r}(e^{[\Gamma]}(\nu_\mu^{[\mc{P}]})) = \min_{d^{[\Gamma]} \leq q^{[\Gamma]}, d^{\l\Gamma\r} = \nu_\mu^\Gamma} \dot\lambda_{ij}^{\l\Gamma\r}(d^{[\Gamma]}).
	\end{equation}
	By concatenation, this guarantees that $\dot\lambda_{ij}^{\l\Omega\r}(e^{[\Omega]}(\nu_\mu^{[\mc{P}]})) \leq \dot\lambda_{ij}^{\l\Omega\r}(d_\mu^{[\Omega]}) \leq \dot{c}_{ij}$ for all $i,j \in \mb{N}$. Then, according to Theorem~\ref{T:schedulability}, $e^{[\Omega]}(\nu_\mu^{[\mc{P}]})$ is feasible, and $e^{[\Omega]}(\nu_\mu^{[\mc{P}]}) \in \mb{F}(\nu_\mu^{[\mc{P}]})$.	
\end{proof}

\begin{proof}[Proof of Theorem~\ref{T:HMSSCond}]
	The condition's necessity follows directly from our argument leading to this theorem, so we need only show its sufficiency. Consider first, the case that $\nu_\mu^{[\mc{P}]}$ is a vertex of $\mb{P}(\beta_\mu^\mc{P})$. Recall, from Section~\ref{SS:polymatroid}, that each vertex of $\mb{P}(\beta_\mu^\mc{P})$ can be identified by a permutation. A permutation over $\mc{P}$ is a bijective map, $\sigma: \mc{P} \rightarrow \{1, 2, \ldots, |\mc{P}|\}$. The vertex identified by $\sigma$, $v_\sigma^{[\mc{P}]}(\beta_\mu^\mc{P})$, is the unique solution to the system of linear equations defined by
	\begin{equation}\label{E:vertexSumExtEntries}
		\textup{(a)}\blank v_\sigma^{\l\mc{S}_\sigma^j\r} (\beta_\mu^\mc{P}) = \beta_\mu^\mc{P}(\mc{S}_\sigma^j), ~j = 0,1,2,\ldots,|\mc{P}|,
		\blank\textup{with}\blank
		\textup{(b)}\blank \mc{S}_\sigma^j := \{\Gamma\in\mc{P} | \sigma(\Gamma) \leq j\},
	\end{equation}
	so for each $\Gamma\in\mc{P}$,
	\begin{equation}\label{E:vertexExt}
		v_\sigma^\Gamma(\beta_\mu^\mc{P}) = \beta_\mu^\mc{P}(\mc{S}_\sigma^{\sigma(\Gamma)}) - \beta_\mu^\mc{P}(\mc{S}_\sigma^{\sigma(\Gamma)-1}).
	\end{equation}
	
	Now, construct $\pi_\sigma$, a permutation over $\Omega$, such that for all $\Gamma,\Gamma'\in\mc{P}$, $\omega\in\Gamma$ and $\omega'\in\Gamma'$, $\pi_\sigma(\omega) < \pi_\sigma(\omega')$ if $\sigma(\Gamma) < \sigma(\Gamma')$. Intuitively, $\pi_\sigma$ can be constructed by first sorting all $\Gamma\in\mc{P}$ according to $\sigma$ and then expanding each $\Gamma$ to its constituent elements. Denoting $\Gamma_{\mc{S}_\sigma^j}$ by $\Gamma_\sigma^j$, implying that $\Gamma_\sigma^j = \bigcup_{\Gamma\in\mc{S}_\sigma^j} \Gamma$, it is immediate from (\ref{E:vertexSumExtEntries})-(b) and (\ref{E:vertexSumEntries})-(b) that $\Gamma_\sigma^j = \Gamma_{\pi_\sigma}^{|\Gamma_\sigma^j|}$. Then, using (\ref{E:vertexSumEntries})-(a), and the definition of $\beta_\mu^\mc{P}$ in (\ref{E:permuBetaPEntries})-(a), we have
	\[
	v_{\pi_\sigma}^{\l\Gamma_\sigma^j\r}(\beta_\mu) = v_{\pi_\sigma}^{\l\Gamma_{\pi_\sigma}^{|\Gamma_\sigma^j|}\r}(\beta_\mu) = \beta_\mu(\Gamma_{\pi_\sigma}^{|\Gamma_\sigma^j|}) = \beta_\mu(\Gamma_\sigma^j) = \beta_\mu^\mc{P}(\mc{S}_\sigma^j),
	\]
	so for each $\Gamma\in\mc{P}$, according to (\ref{E:vertexExt}),
	\begin{equation}\label{E:vertexRelation}
		v_{\pi_\sigma}^{\l\Gamma\r}(\beta_\mu) = v_{\pi_\sigma}^{\l\Gamma_\sigma^{\sigma(\Gamma)}\r}(\beta_\mu) - v_{\pi_\sigma}^{\l\Gamma_\sigma^{\sigma(\Gamma)-1}\r}(\beta_\mu) = \beta_\mu^\mc{P}(\mc{S}_\sigma^{\sigma(\Gamma)}) - \beta_\mu^\mc{P}(\mc{S}_\sigma^{\sigma(\Gamma)-1}) = v_\sigma^\Gamma(\beta_\mu^\mc{P}),
	\end{equation}
	implying that $v_{\pi_\sigma}^{[\Omega]}(\beta_\mu) \in \mb{H}(v_\sigma^{[\mc{P}]}(\beta_\mu^\mc{P}))$. According to Theorem \ref{T:FP}, $v_{\pi_\sigma}^{[\Omega]}(\beta_\mu) \in \mb{F}_\mu$. It follows that $v_{\pi_\sigma}^{[\Omega]}(\beta_\mu) \in \mb{F}(v_\sigma^{[\mc{P}]}(\beta_\mu^\mc{P}))$, $\mb{F}(v_\sigma^{[\mc{P}]}(\beta_\mu^\mc{P}))$ is non-empty, and $v_\sigma^{[\mc{P}]}(\beta_\mu^\mc{P})$ is feasible.
	
	Consider next, a general $\nu_\mu^{[\mc{P}]} \in \mb{P}(\beta_\mu^\mc{P})$. Let $\Pi^\mc{P}$ be the set of all permutations over $\mc{P}$. As $\nu_\mu^{[\mc{P}]}$ is a convex combination of $\mb{P}(\beta_\mu^\mc{P})$'s vertices, there must exist, for all $\sigma \in \Pi^\mc{P}$, $w_\sigma \in [0, 1]$ in $\mb{R}$, with $\sum_{\sigma\in\Pi^\mc{P}} w_\sigma = 1$, such that
	\[
	\nu_\mu^{[\mc{P}]} = \sum_{\sigma\in\Pi^\mc{P}} w_\sigma v_\sigma^{[\mc{P}]}(\beta_\mu^\mc{P}).
	\]
	Now construct
	\[
	d_\mu^{[\Omega]} = \sum_{\sigma\in\Pi^\mc{P}} w_\sigma v_{\pi_\sigma}^{[\Omega]}(\beta_\mu).
	\]
	On the one hand, it is immediate from (\ref{E:vertexRelation}) that, for each $\Gamma\in\mc{P}$,
	\[
	d_\mu^{\l\Gamma\r} = \sum_{\sigma\in\Pi^\mc{P}} w_\sigma v_{\pi_\sigma}^{\l\Gamma\r}(\beta_\mu) = \sum_{\sigma\in\Pi^\mc{P}} w_\sigma v_\sigma^\Gamma(\beta_\mu^\mc{P}) = \nu_\mu^\Gamma,
	\]
	implying that $d_\mu^{[\Omega]} \in \mb{H}(\nu_\mu^{[\mc{P}]})$. On the other hand, since $v_{\pi_\sigma}^{[\Omega]}(\beta_\mu) \in \mb{F}_\mu$, $d_\mu^{[\Omega]} \in \mb{F}_\mu$. It follows that $d_\mu^{[\Omega]} \in \mb{F}(\nu_\mu^{[\mc{P}]})$, $\mb{F}(\nu_\mu^{[\mc{P}]})$ is non-empty, and $\nu_\mu^{[\mc{P}]}$ is feasible.
	
	One subtlety does need to be addressed. Close examination of the preceding proof reveals that what is proved is that $\mb{F}(\nu_\mu^{[\mc{P}]})$ is non-empty in $\mb{R}^n$, not that it contains any integral point. This, however, is not an obstacle to the joint application of this theorem and Theorem~\ref{T:HMSS}. In fact, nowhere is the integrality of $\mb{F}(\nu_\mu^{[\mc{P}]})$ ever used in Theorem~\ref{T:HMSS}'s proof, so it holds as long as $\mb{F}(\nu_\mu^{[\mc{P}]})$ is non-empty in $\mb{R}^n$. In light of this observation, $\mb{F}(\nu_\mu^{[\mc{P}]})$'s non-emptiness in $\mb{R}^n$ is not only necessary but also sufficient for its integrality. If $\mb{F}(\nu_\mu^{[\mc{P}]})$ is non-empty in $\mb{R}^n$, Theorem~\ref{T:HMSS} guarantees that $e^{[\Omega]}(\nu_\mu^{[\mc{P}]}) \in \mb{F}(\nu_\mu^{[\mc{P}]})$. But, by definition, $e^{[\Omega]}(\nu_\mu^{[\mc{P}]})$ is integral. So $\mb{F}(\nu_\mu^{[\mc{P}]})$ contains at least one integral point.
\end{proof}	

\section{Proofs for Section \ref{S:MPS}}\label{A:MPSProof}

\begin{proof}[Proof of Theorem~\ref{T:SSig}]
	To prove this theorem, we introduce $\bs\varepsilon = [\varepsilon_j]_{j\in\mb{N}}:=[0, \infty, \infty, \ldots] \in \mb{U}$, that is, $\varepsilon_0 := 0$ and $\varepsilon_j := \infty$ for all $j>0$. Intuitively, $\bs\varepsilon$ models the arrival of an infinite burst of tasks in slot~$t$. Using (\ref{E:psiS}) and (\ref{E:SMDefEntries})-(c), we have
	\begin{equation}\label{E:epsilon++}
		\psi_j^S(\mc{R}^i\bs\varepsilon+b\bs\delta) = \min_{k\in\mb{N}} ([\mc{R}^i\bs\varepsilon]_k+b\delta_k+s_{kj}) = \min_{k \leq i} (b\delta_k+s_{kj}) = \min \{s_{0j}, b\delta_i+s_{ij}\},          
	\end{equation}
	where the second equality holds because, by definition, $[\mc{R}^i\bs\varepsilon]_k = 0$ if $k \leq i$ and $[\mc{R}^i\bs\varepsilon]_k = \infty$ if $k > i$. Then, on the one hand, since by default, $\mc{R}^i\bs\varepsilon+b\bs\delta \in \mb{U}\ua b$, using (\ref{E:signature}), (\ref{E:epsilon++}), and (\ref{E:SMDefEntries})-(d), we have
	\[
	\begin{IEEEeqnarraybox}[][c]{rCl}
		\lambda_{ij}(\bs\psi^S) \geq (\psi_j^S(\bs{q})-q_i)^+|_{\bs{q}=\mc{R}^i\bs\varepsilon+b\bs\delta}
		& =    & (\psi_j^S(\mc{R}^i\bs\varepsilon+b\bs\delta)-[\mc{R}^i\bs\varepsilon]_i-b\delta_i)^+ \\
		& =    & (\min\{s_{0j}, b\delta_i+s_{ij}\}-b\delta_i)^+ \\
		& =    & \min\{(s_{0j}-b\delta_i)^+, s_{ij}\} \\
		& =    & s_{ij}.
	\end{IEEEeqnarraybox}
	\]
	On the other hand, using (\ref{E:signature}) and (\ref{E:psiS}), we have
	\[
	\lambda_{ij}(\bs\psi^S) = \max_{\bs{q}\in\mb{U}\ua b} \left(\min_{k\in\mb{N}} (q_k+s_{kj})-q_i\right)^+
	\leq \max_{\bs{q}\in\mb{U}\ua b} \left(q_i+s_{ij}-q_i\right)^+ =  s_{ij}.
	\]
	It follows that $\lambda_{ij}(\bs\psi^S) = s_{ij}$.
\end{proof}

\begin{proof}[Proof of Theorem~\ref{T:SCondSig}]
	In the case that $i=0$, on the one hand, since by default, $\mc{R}\bs\varepsilon+q\bs\delta \in \mb{U}|q$, using (\ref{E:CondSig}) and (\ref{E:epsilon++}), we have
	\[
	\hat{s}_{0j} \geq \psi_j^S(\bs{q})|_{\bs{q}=\mc{R}\bs\varepsilon+q\bs\delta} = \psi_j^S(\mc{R}\bs\varepsilon+q\bs\delta) = \min \{s_{0j}, q+s_{1j}\}.
	\]
	On the other hand, using (\ref{E:CondSig}) and (\ref{E:psiS}), we have
	\[
	\hat{s}_{0j} = \max_{\bs{q}\in\mb{U}|q} \min_{k\in\mb{N}} (q_k+s_{kj}) \leq \max_{\bs{q}\in\mb{U}|q} \min \{s_{0j}, q_1+s_{1j}\} = \min \{s_{0j}, q+s_{1j}\}.
	\]
	It follows that (\ref{E:SCondSigCaseA}) holds.
	
	In the case that $i>0$, on the one hand, since by default, $\mc{R}^i\bs\varepsilon+q\bs\delta \in \mb{U}|q$, using (\ref{E:CondSig}) and (\ref{E:epsilon++}), we have
	\[
	\begin{IEEEeqnarraybox}[][c]{rCl}
		\hat{s}_{ij} \geq (\psi_j^S(\bs{q})-q_i)^+|_{\bs{q}=\mc{R}^i\bs\varepsilon+q\bs\delta}
		& =    & (\psi_j^S(\mc{R}^i\bs\varepsilon+q\bs\delta)-[\mc{R}^i\bs\varepsilon]_i-q)^+ \\
		& =    & (\min \{s_{0j}, q+s_{ij}\}-q)^+ \\
		& =    & \min \{(s_{0j}-q)^+, s_{ij}\},
	\end{IEEEeqnarraybox}
	\]
	where the second equality holds because, by definition, $[\mc{R}^i\bs\varepsilon]_i = 0$. On the other hand, using (\ref{E:CondSig}) and (\ref{E:psiS}), we have
	\[
	\begin{IEEEeqnarraybox}[][c]{rCl}
		\hat{s}_{ij}  = \max_{\bs{q}\in\mb{U}|q} \left(\min_{k\in\mb{N}} (q_k+s_{kj})-q_i\right)^+
		& \leq & \max_{\bs{q}\in\mb{U}|q} (\min \{s_{0j}, q_i+s_{ij}\}-q_i)^+ \\
		& =    & \max_{\bs{q}\in\mb{U}|q} \min \{(s_{0j}-q_i)^+, s_{ij}\} \\
		& \leq & \min \{(s_{0j}-q)^+, s_{ij}\},
	\end{IEEEeqnarraybox}
	\]
	where the final inequality holds because, as $i>0$, $q_i \geq q_1 = q$ for all $\bs{q}\in\mb{U}|q$. It follows that (\ref{E:SCondSigCaseB}) holds.
\end{proof}

\begin{proof}[Proof of Theorem~\ref{T:MtoS}]
	It is easy to verify that $S$ is indeed a spectral matrix. Notice that (\ref{E:MtoSCases}) can be rewritten as
	\[
		s_{ij} = \min \left\{(m_{0j}-b\delta_i)^+, \min_{k \leq i} m_{kj}\right\}.
	\]
	So, for all $\bs{q}\in\mb{U}\ua b$ and $j\in\mb{N}$, according to (\ref{E:psiS}),
	\[
	\psi_j^S(\bs{q}) = \min \left\{\min_{i \in \mb{N}} [q_i+(m_{0j}-b\delta_i)^+], \min_{i \in \mb{N}} \min_{k \leq i} (q_i+m_{kj})\right\}.
	\]
	On the one hand, for the left term of the outer minimization, using (\ref{E:vecArrival}) and (\ref{E:psiM+}), we have
	\[
	\min_{i \in \mb{N}} [q_i+(m_{0j}-b\delta_i)^+] \geq \min_{i \in \mb{N}} (q_i-b\delta_i)+m_{0j} = \min_{i \in \mb{N}} a_i+m_{0j} = m_{0j} \geq \psi_j^M(\bs{q}),
	\]
	where the final inequality holds because $m_{0j} = q_0+m_{0j} \geq \psi_j^M(\bs{q})$. On the other hand, for the right term, reversing the order of the two $\min$ operators, and using (\ref{E:psiM+}), we have
	\[
	\min_{i \in \mb{N}} \min_{k \leq i} (q_i+m_{kj}) = \min_{k \in \mb{N}} \min_{i \geq k} (q_i+m_{kj}) = \min_{k \in \mb{N}} (q_k+m_{kj}) = \psi_j^M(\bs{q}),
	\]
	where the second equality holds because $q_i \geq q_k$ for all $i \geq k$. It follows that $\psi_j^S(\bs{q}) = \psi_j^M(\bs{q})$, so $\bs\psi^S = \bs\psi^M$.
\end{proof}

\begin{proof}[Proof of Theorem~\ref{T:MUpdate}]
	We first show that $\dot{M}$ is a cumulative matrix. Using (\ref{E:MUpdateCases}) and (\ref{E:CMDefEntries}), it is easy to verify that $\dot{m}_{ij} \leq \dot{m}_{i, j+1}$ for all $i,j \in \mb{N}$, and that $\dot{m}_{ij} = 0$ if $i>0$ and $i \geq j$. So we need only show that $\dot{m}_{00} = 0$. Notice that, since $d \geq p$, according to (\ref{E:WCSImmediate}), (\ref{E:psiM+}), and (\ref{E:CMDefEntries})-(a), we have
	\[
	d \geq \max_{\bs{q} \in \mb{U}|q} \psi_1^M(\bs{q}) = \max_{\bs{q} \in \mb{U}|q} \min \{m_{01}, q_1+m_{11}\}
	= \min \{m_{01}, q\}.
	\]
	So, according to (\ref{E:MUpdateCaseA}) and (\ref{E:CMDefEntries})-(a),
	\[
	\dot{m}_{00} = (\min \{m_{01}, q+m_{11}\}-d)^+= (\min \{m_{01}, q\}-d)^+ = 0.
	\]
	
	It remains to show that $\bs{\dot\psi}$ can be identified by $\dot{M}$. For all $\bs{\dot{q}} \in \mb{U}\ua\dot{b}$ and $j \in \mb{N}$, using (\ref{E:WCSUpdateB}) and (\ref{E:psiM+}), we have
	\[
	\dot\psi_j(\bs{\dot{q}}) = [\mc{R}^{-1}(\bs\psi^M(\bs{q})-d\bs\delta)^+]_j = (\psi_{j+1}^M(\bs{q})-d)^+ = \left(\min_{i \in \mb{N}} (q_i+m_{i, j+1})-d\right)^+,
	\]
	so
	\[
	\begin{IEEEeqnarraybox}[][c]{rCl}
		\dot\psi_j(\bs{\dot{q}})
		&=& \left(\min \left\{m_{0, j+1}, q_1+m_{1, j+1}, \min_{i>1} (q_i+m_{i, j+1})\right\}-d\right)^+ \\
		&=& \left(\min \left\{m_{0, j+1}, q_1+m_{1, j+1}, \min_{i>0} (q_{i+1}+m_{i+1, j+1})\right\}-d\right)^+.
	\end{IEEEeqnarraybox}
	\]
	Then, according to (\ref{E:qUpdateCaseB}) and (\ref{E:MUpdateCases}),
	\[
	\begin{IEEEeqnarraybox}[][c]{rCl}
		\dot\psi_j(\bs{\dot{q}})
		&=& \left(\min \left\{m_{0, j+1}, q+m_{1, j+1}, \min_{i>0} (\dot{q}_i+d+m_{i+1, j+1})\right\}-d\right)^+ \\
		&=& \min \left\{(\min\{m_{0, j+1}, q+m_{1, j+1}\}-d)^+, \min_{i>0} (\dot{q}_i+m_{i+1, j+1})\right\} \\
		&=& \min_{i \in \mb{N}}(\dot{q}_i+\dot{m}_{ij}).
	\end{IEEEeqnarraybox}
	\]
	It follows that $\bs{\dot\psi}(\bs{\dot{q}}) = \bs{\dot{q}} \otimes \dot{M}$, so $\bs{\dot\psi} = \bs\psi^{\dot{M}}$.
\end{proof}

\section{Performance Bounds of Worst-Case Services}\label{A:performance}

In \cite{Cruz:1991A, Cruz:1991B, Parekh:1993, Parekh:1994}, cumulative curve models were shown to be useful tools for performance analysis. The same methods can be applied here given our cumulative vector models. Consider first flow backlogs. By definition, the number of tasks left unserved prior to slot $t+j$ is
\begin{equation}\label{E:bj}
	b_j := q_j-d_j, \blank\blank j = 1, 2, 3, \ldots,
\end{equation}
according to which, $b_1 = \dot{b}$, $b_2 = \ddot{b}$, and so on. As illustrated in Fig.~\ref{F:Bounds}, if the flow is guaranteed $\bs\psi$, this backlog is bounded by
\begin{equation}\label{E:distV}
	b_j \leq q_j-\psi_j(\bs{q}) \leq \max_{j > 0} (q_j-\psi_j(\bs{q})).
\end{equation}
Consider next task delays. Recall that, according to (\ref{E:tau}), the $h$th task in $\bs{d}$ is served in slot $t+\tau_h(\bs{d})$. Notice that this service corresponds to the $h$th task in $\bs{q}$, not $\bs{a}$, because the $b$ tasks left unserved prior to slot $t$ have to be served before any new arrival. One subtlety here is that there is not enough information to determine when the $b$ tasks arrived. However, if we disregard delays experienced prior to slot~$t$ and treat the $b$ tasks as if they are new arrivals, the $h$th task in $\bs{d}$ can be viewed as arriving in slot $t+\tau_h(\bs{q})$, and its delay is
\begin{equation}\label{E:thetah}
	\theta_h := \tau_h(\bs{d})-\tau_h(\bs{q}), \blank\blank h = 1, 2, 3, \ldots.
\end{equation}
Then, as illustrated in Fig.~\ref{F:Bounds}, if the flow is guaranteed $\bs\psi$, this delay is bounded by
\begin{equation}\label{E:distH}
	\theta_h \leq \tau_h(\bs\psi(\bs{q}))-\tau_h(\bs{q}) \leq \max_{h>0} (\tau_h(\bs\psi(\bs{q}))-\tau_h(\bs{q})).
\end{equation}

Notice that the bounds in (\ref{E:distV}) and (\ref{E:distH}) are tied to a specific $\bs{q}$. Confining $\bs{q}$ to a subset of $\mb{U} \ua b$ so that the worst case of the worst-case bounds can be determined, or endowing $\mb{U} \ua b$ with a probability measure so that distributions on the bounds can be derived, extends the bounds to cases when $\bs{q}$ is uncertain. Examples of the former approach can be found in \cite{Cruz:1991A, Cruz:1991B, Parekh:1993, Parekh:1994}, while an example of the latter approach can be found in \cite{Sidi:1993}, where tasks arrive according to a {\it Poisson} process and are served according to a leaky-bucket scheme.

Not only can we bound the performance guaranteed by given worst-case services, we can also design worst-case services that guarantee given bounds. In the case of service curves, which, according to Section \ref{SS:VPS}, are special cases of worst-case services, approaches to designing service curves to guarantee given bounds are well documented \cite{Chang:2000, Boudec:2001, Bouillard:2018}. More generally, using (\ref{E:distV}) and (\ref{E:distH}), we can design worst-case services that guarantee backlog bounds slot-by-slot and delay bounds task-by-task. In the simplest cases, we set
\begin{equation}\label{E:CBS}
	\bs\psi(\bs{q}) = (\bs{q}-\bar{b}\bs\delta)^+ \blank\blank\forall \bs{q} \in \mb{U} \ua b,
\end{equation}
to guarantee that all flow backlogs are uniformly bounded by $\bar{b}$, and
\begin{equation}\label{E:CDS}
	\bs\psi(\bs{q}) = \mc{R}^{\bar\theta}\bs{q} \blank\blank\forall \bs{q} \in \mb{U} \ua b,
\end{equation}
to guarantee that all task delays are uniformly bounded by $\bar\theta$. Since both cases cover the entirety of $\mb{U} \ua b$, neither can be guaranteed by a finite-capacity server because, for instance, it is impossible to guarantee the service in (\ref{E:CBS}) if $\bs{q} \geq (\bar{b}+c+1)\bs\delta$ or the service in (\ref{E:CDS}) if $\bs{q} \geq (\bar\theta c + 1) \bs\delta$. However, in both cases, this issue can be avoided by confining coverage to a judiciously selected subset of $\mb{U} \ua b$.

\section{Interpreting Theorem \ref{T:SigUpdate}}\label{A:SigUpdateInterpret}

The spectral update relation in (\ref{E:SigUpdateCases}) has a straightforward interpretation. By definition, both $\dot\lambda_{ij}$ and $\hat\lambda_{i+1, j+1}$ specify the least capacity that must be reserved during interval $[t+i+1, t+j+1)$. The difference is that $\hat\lambda_{i+1, j+1}$ is computed relative to slot~$t$, {\it before} any task has been served, by letting $d_{i+1} = q_{i+1}$ for each $\bs{q} \in \mb{U} | q$, whereas $\dot\lambda_{ij}$ is computed relative to $t+1$, {\it after} the service of $d$ tasks, by letting $\dot{d}_i = \dot{q}_i$ for each $\bs{\dot{q}} \in \mb{U} \ua \dot{b}$. But, given any $d \leq q$, according to (\ref{E:vecqUpdate}), $\bs{\dot{q}} \in \mb{U} \ua \dot{b}$ is equivalent to $\bs{q} \in \mb{U} | q$, and moreover, if $i > 0$, according to (\ref{E:vecdUpdate}) and (\ref{E:qUpdateCaseB}), $\dot{d}_i = \dot{q}_i$ implies that $d_{i+1} = \dot{d}_i+d = \dot{q}_i+d = q_{i+1}$. So $\dot\lambda_{ij}=\hat\lambda_{i+1, j+1}$ in the case that $i>0$.

In the case that $i=0$, {\it pre}-service, the least capacity that must be reserved during interval $[t+1, t+j+1)$ is still $\hat\lambda_{1, j+1}$. It is computed by letting $d_1 = q_1$, that is, $d=q$. But, {\it post}-service, it cannot always remain $\hat\lambda_{1, j+1}$, because $d \leq q$ tasks have been actually served. According to (\ref{E:SigUpdateCaseA}), it now becomes $(\hat\lambda_{0, j+1}-d)^+$, the least capacity that must be reserved during $[t,t+j+1)$ minus the reserve released when the $d$ tasks are served. This is consistent with the fact that, according to (\ref{E:lambdaBound}), $\dot{\lambda}_{0j} \geq \hat\lambda_{1, j+1}$.

\section{Multiplexing Gains of Worst-Case Systems}\label{A:gains}

By definition, the schedulability condition distinguishes what can be guaranteed from what cannot. In particular, (\ref{E:schedulability}) can be rewritten as
\begin{equation}\label{E:rhoSchedulability}
	c \geq \rho(\Omega) := \max_{i<j} \frac{\lambda_{ij}^{\l\Omega\r}}{j-i}~.
\end{equation}
That is to say, if $c < \rho(\Omega)$, no matter how smart the scheduler is, it is impossible to simultaneously guarantee $\bs\psi^\omega$ for all $\omega\in\Omega$. Extending the definition of $\rho(\Omega)$ in (\ref{E:rhoSchedulability}) to all $\Gamma\subseteq\Omega$ so that
\begin{equation}\label{E:rho}
	\rho(\Gamma):= \max_{i<j} \frac{\lambda_{ij}^{\l\Gamma\r}}{j-i}~,
\end{equation}
enables finer distinctions to be made. By construction, if all flows in $\Gamma$ are served by a single server, $\rho(\Gamma)$ is the least capacity required by this server to simultaneously guarantee $\bs\psi^\omega$ for all $\omega\in\Gamma$. If each flow is served by a separate server, then, $\sum_{\omega\in\Omega} \rho(\{\omega\})$ is the least total capacity required by these servers. Let
\begin{equation}\label{E:eta}
	\eta:= \frac{\sum_{\omega\in\Omega} \rho(\{\omega\})}{\rho(\Omega)}~.
\end{equation}
Clearly, the larger the $\eta$, the greater capacity utilization that can be achieved through multiplexing. For this reason, we call $\eta$ the {\it \textbf{multiplexing gain}} of worst-case system $\bs\psi^{[\Omega]}$.

For all $\Gamma,\Gamma'\subseteq\Omega$, observe that
\begin{equation}\label{E:rhoProp}
	\rho(\Gamma+\Gamma') \leq \rho(\Gamma)+\rho(\Gamma'),
\end{equation}
because
\[
	\max_{i<j} \frac{\lambda_{ij}^{\l\Gamma+\Gamma'\r}}{j-i} = \frac{\lambda_{i_*j_*}^{\l\Gamma+\Gamma'\r}}{j_*-i_*} \leq \frac{\lambda_{i_*j_*}^{\l\Gamma\r}}{j_*-i_*} + \frac{\lambda_{i_*j_*}^{\l\Gamma'\r}}{j_*-i_*}\\
	\leq \max_{i<j} \frac{\lambda_{ij}^{\l\Gamma\r}}{j-i}+\max_{i<j} \frac{\lambda_{ij}^{\l\Gamma'\r}}{j-i}~,
\]
where $i_*$ and $j_*$ maximize $\frac{\lambda_{ij}^{\l\Gamma+\Gamma'\r}}{j-i}$. It follows from (\ref{E:eta}) and (\ref{E:rhoProp}) that $\eta \geq 1$. Moreover, it is easy to see that $\eta=1$ if for all $\omega,\omega'\in\Omega$, $\lambda_{ij}^\omega \propto \lambda_{ij}^{\omega'}$, while $\eta$ increases as the correlation among the $\lambda_{ij}^\omega$'s decreases. Roughly speaking, the more diverse the guarantees required by $\bs\psi^{[\Omega]}$, the larger the multiplexing gain.

As the number of flows grows, the complexity of many scheduling policies grows so fast that, even when $\eta$ is large enough to justify multiplexing, divide-and-conquer schemes remain attractive. To be precise, let $\mc{P} \subseteq 2^\Omega$ be a partition of $\Omega$. If, for each $\Gamma \in \mc{P}$, all flows in $\Gamma$ are served by a separate server, then, $\sum_{\Gamma\in\mc{P}} \rho(\Gamma)$ is the least total capacity required by these servers. Paralleling (\ref{E:eta}), let
\begin{equation}\label{E:etaP}
	\eta^\mc{P}:= \frac{\sum_{\omega\in\Omega} \rho(\{\omega\})}{\sum_{\Gamma\in\mc{P}} \rho(\Gamma)}~.
\end{equation}
Using (\ref{E:rhoProp}) and (\ref{E:eta}), it is easy to verify that $1 \leq \eta^\mc{P} \leq \eta$. Now, if a $\mc{P}$ such that $\max_{\Gamma \in \mc{P}} |\Gamma|$ is sufficiently small and $\eta^\mc{P}$ is sufficiently close to $\eta$ can be identified, then scheduling complexity can be reduced without sacrificing much of the multiplexing gain.

\section{The Baseline Permutohedron}\label{A:permuBase}

An alternative, and complementary, approach to selecting feasible schedules is motivated by the following observation. In the case that $\mu = \beta(\Omega)$, using (\ref{E:betaMu}), it is easy to verify that $\beta_\mu = \beta$, that is, $\beta_\mu(\Gamma) = \beta(\Gamma)$ for all $\Gamma\subseteq\Omega$, because, according to (\ref{E:betaProp+}), $\beta(\Omega) - \beta(\Gamma) \leq q^{\l\overline\Gamma\r}$. So, in this case, $\mb{P}(\beta_\mu) = \mb{P}(\beta)$. We call $\mb{P}(\beta)$ the {\it \textbf{baseline permutohedron}}. It is the bottom face of $\mb{F}$ in the sense that $d^{\l\Omega\r}$ is minimized in $\mb{F}$ if $d^{[\Omega]} \in \mb{P}(\beta)$. In Fig.~\ref{F:FeasibleRegion}, $\mb{P}(\beta)$ is line segment $\overline{AB}$.

To use $\mb{P}(\beta)$ to select a feasible schedule, first select $d_*^{[\Omega]}$ from $\mb{P}(\beta)$, and then select $d^{[\Omega]}$ such that $d_*^{[\Omega]} \leq d^{[\Omega]} \leq q^{[\Omega]}$ and $d^{\l\Omega\r} \leq c$. In the first step, $d_*^{[\Omega]}$ can be freely selected to be, for instance, a vertex, the vertex centroid, the max-slack schedule, or a per-class max-slack schedule. In the second step, significantly, nothing is required but a free allocation of excess capacity, $c-\beta(\Omega)$, subject to the causality constraint, $d^{[\Omega]} \leq q^{[\Omega]}$. So, for instance, a generalized-processor-sharing policy, as analyzed in \cite{Parekh:1993, Parekh:1994}, can be used to allocate this excess according to any preset ratios, or even ratios set dynamically by $d_*^{[\Omega]}$. This usage of excess capacity was also proposed in \cite{Stoica:2000} to improve the EDF scheduler's fairness when it is used to guarantee service curves.

\section{Selecting Per-Class Max-Slack Schedules}\label{A:selectPerClass}

Recall from the last paragraph of Section~\ref{SS:FP} that, by selecting a vertex of $\mb{P}(\beta_\mu)$, we can serve flows according to a priority order. Similarly, by selecting a vertex of $\mb{P}(\beta_\mu^\mc{P})$, $v_\sigma^{[\mc{P}]}(\beta_\mu^\mc{P})$, where $\sigma$ is a permutation over $\mc{P}$ and $v_\sigma^{[\mc{P}]}(\beta_\mu^\mc{P})$ can be identified by (\ref{E:vertexExt}), and thus
\begin{equation}\label{E:HMSSvertex}
	e_\sigma^{[\Omega]}(\beta_\mu^\mc{P}) := e^{[\Omega]}(v_\sigma^{[\mc{P}]}(\beta_\mu^\mc{P})),
\end{equation}
an inter-class priority order is enforced such that the larger $\sigma(\Gamma)$ is, the higher the priority that class~$\Gamma$ enjoys. For instance, in the case that $\mc{P} = \{\Gamma,\overline\Gamma\}$, according to (\ref{E:permuBetaPEntries}), $\mb{P}(\beta_\mu^\mc{P})$ is a line segment. One vertex, $[\beta_\mu(\Gamma), \mu-\beta_\mu(\Gamma)]$, assigns to class~$\overline\Gamma$ the highest priority possible relative to class~$\Gamma$, and the other, $[\mu-\beta_\mu(\overline\Gamma), \beta_\mu(\overline\Gamma)]$, reverses the priority assignment to favor $\Gamma$.

Recall also that, by selecting the vertex centroid of $\mb{P}(\beta_\mu)$, we can serve flows according to a fairness criterion. Similarly, by selecting the vertex centroid,
\begin{equation}\label{E:u-f}
	v_\text{F}^{[\mc{P}]}(\beta_\mu^\mc{P}) = \frac{1}{|\mc{P}|!} \sum_{\sigma\in\Pi^\mc{P}} v_\sigma^{[\mc{P}]}(\beta_\mu^\mc{P}),
\end{equation}
and thus
\begin{equation}\label{E:v-f}
	e_\text{F}^{[\Omega]}(\beta_\mu^\mc{P}) := e^{[\Omega]}(v_\text{F}^{[\mc{P}]}(\beta_\mu^\mc{P})),
\end{equation}
where $v_\text{F}^{[\mc{P}]}(\beta_\mu^\mc{P})$ can be rounded if it is not an integral point, an inter-class fairness criterion is enforced. But this fairness is not intra-class. In fact, a flow can still be starved by other flows in the same class. To address this concern, an alternative approach is inspired by the multi-level feedback queue widely used in operating systems \cite{ArpaciDusseau:2018} (ch. 8). The idea is to fix an inter-class priority order, but allow each flow's priority to be adjusted dynamically. In particular, if a flow has been starved, it will be moved to a higher-priority class, while if a flow has been served more than adequately, it will be moved to a lower-priority class.

\section{Composing Dual-Curve Services}\label{A:composition}

A worst-case service, $\bs\psi$, is {\it \textbf{monotone}} if $\bs{q}\geq\bs{q'}$ implies that $\bs\psi(\bs{q}) \geq \bs\psi(\bs{q'})$. Intuitively monotonicity ensures that increasing the number of task arrivals never degrades service. The next theorem establishes that monotone services are composable.

\begin{theorem}\label{T:composability}
	As illustrated in Fig.~\ref{F:Composability}, when a flow is guaranteed monotone services $\bs\psi^\textup{I}$ and $\bs\psi^\textup{II}$ by two servers operating in series, the effective service can be modeled by a single server, with $\bs{a}= \bs{a}^\textup{I}$, $b = b^\textup{I}+b^\textup{II}$, and $\bs{d}= \bs{d}^\textup{II}$, that guarantees monotone service $\bs\psi= \bs\psi^\textup{II} \star \bs\psi^\textup{I}$, where
	\begin{equation}\label{E:composability}
		\bs\psi^\textup{II} \star \bs\psi^\textup{I}(\bs{q}) := \bs\psi^\textup{II}(\bs\psi^\textup{I}(\bs{q}-b^{\textup{II}}\bs\delta)+b^{\textup{II}}\bs\delta) \blank\blank \forall \bs{q} \in \mb{U} \ua b.
	\end{equation}
\end{theorem}

\begin{figure}[t]
	\centering \scalebox{1.000}{\includegraphics{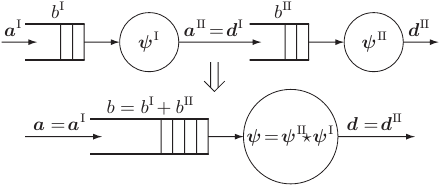}} \caption{When a flow is guaranteed monotone services $\bs\psi^\textup{I}$ and $\bs\psi^\textup{II}$ by two servers operating in series, the effective service can be modeled by a single server, with $\bs{a}= \bs{a}^\textup{I}$, $b = b^\textup{I}+b^\textup{II}$, and $\bs{d}= \bs{d}^\textup{II}$, that guarantees monotone service $\bs\psi= \bs\psi^\textup{II} \star \bs\psi^\textup{I}$. Notice that, as reflected by the relation that $\bs{a}^\textup{II} = \bs{d}^\textup{I}$, we assume that a task can be served by the first server and then the second server in the same slot.}
	\label{F:Composability}
	\Description{Composing two monotone services into one.}
\end{figure}

\begin{proof}
	Since $\bs\psi^\text{II}$ is monotone, using (\ref{E:WCSDef}) twice, we have
	\[
		\bs\psi(\bs{q}) = \bs\psi^\text{II}(\bs\psi^\text{I}(\bs{q}-b^{\text{II}}\bs\delta)+b^{\text{II}}\bs\delta) \leq \bs\psi^\text{II}(\bs{q}-b^{\text{II}}\bs\delta+b^{\text{II}}\bs\delta) = \bs\psi^\text{II}(\bs{q}) \leq \bs{q}.
	\]
	So, according to Definition~\ref{D:WCS}, $\bs\psi$ is a worst-case service. In addition, since both $\bs\psi^\text{I}$ and $\bs\psi^\text{II}$ are monotone, if $\bs{q}\geq\bs{q'}$, it is immediate that
	\[
		\bs\psi(\bs{q}) = \bs\psi^\text{II}(\bs\psi^\text{I}(\bs{q}-b^{\text{II}}\bs\delta)+b^{\text{II}}\bs\delta) \geq \bs\psi^\text{II}(\bs\psi^\text{I}(\bs{q'}-b^{\text{II}}\bs\delta)+b^{\text{II}}\bs\delta)=\bs\psi(\bs{q'}).
	\]
	Hence, $\bs\psi$ is also monotone.
	
	To show that the flow is guaranteed $\bs\psi$, it is convenient to recast worst-case services as functions of $\bs{a}$ instead of $\bs{q}$. Specifically, for all $\bs{a}^\text{I} \in \mb{U}$, let
	\[
		\bs\varphi^\text{I}(\bs{a}^\text{I}) := \bs\psi^\text{I}(\bs{a}^\text{I}+b^\text{I}\bs\delta) = \bs\psi^\text{I}(\bs{q}^\text{I}) \leq \bs{d}^\text{I},
	\]
	and, for all $\bs{a}^\text{II} \in \mb{U}$, let
	\[
		\bs\varphi^\text{II}(\bs{a}^\text{II}) := \bs\psi^\text{II}(\bs{a}^\text{II}+b^\text{II}\bs\delta) = \bs\psi^\text{II}(\bs{q}^\text{II}) \leq \bs{d}^\text{II}.
	\]
	As $\bs\psi^\text{II}$ is monotone, $\bs\varphi^\text{II}$ is also monotone, so
	\[
		\bs{d} = \bs{d}^\text{II} \geq \bs\varphi^\text{II}(\bs{a}^\text{II}) = \bs\varphi^\text{II}(\bs{d}^\text{I}) \geq \bs\varphi^\text{II}(\bs\varphi^\text{I}(\bs{a}^\text{I})) = \bs\varphi^\text{II}(\bs\varphi^\text{I}(\bs{a})).
	\]
	But, by definition,
	\[
		\bs\varphi^\text{II}(\bs\varphi^\text{I}(\bs{a})) = \bs\psi^\text{II}(\bs\psi^\text{I}(\bs{a}^\text{I}+b^\text{I}\bs\delta)+b^\text{II}\bs\delta) = \bs\psi^\text{II}(\bs\psi^\text{I}(\bs{a}+b\bs\delta-b^{\text{II}}\bs\delta)+b^{\text{II}}\bs\delta) = \bs\psi^\text{II}(\bs\psi^\text{I}(\bs{q}-b^{\text{II}}\bs\delta)+b^{\text{II}}\bs\delta).
	\]
	Therefore, according to (\ref{E:composability}), $\bs{d} \geq \bs\psi(\bs{q})$ for all $\bs{q}\in\mb{U}\ua b$, and $\bs\psi$ is guaranteed.
\end{proof}

Using (\ref{E:psiuv}), it is easy to verify that dual-curve services are monotone. They are also composable as shown by the next theorem.

\begin{theorem}\label{T:DVSComposability}
	In Theorem~\ref{T:composability}, if $\bs\psi^\textup{I}$ and $\bs\psi^\textup{II}$ are both dual-curve services that can be identified by $(\bs{u}^\textup{I}, \bs{v}^\textup{I})$ and $(\bs{u}^\textup{II}, \bs{v}^\textup{II})$, respectively, that is, $\bs\psi^\textup{I} = \bs\psi^{(\bs{u}^\textup{I}, \bs{v}^\textup{I})}$ and $\bs\psi^\textup{II} = \bs\psi^{(\bs{u}^\textup{II}, \bs{v}^\textup{II})}$, then $\bs\psi=\bs\psi^\textup{II}\star\bs\psi^\textup{I}$ is also a dual-curve service that can be identified by $(\bs{u}, \bs{v})$, that is, $\bs\psi = \bs\psi^{(\bs{u}, \bs{v})}$, where, for all $j\in\mb{N}$,
	\begin{equation}\label{E:uvComposing}
		\textup{(a)}\blank u_j = \min \left\{u_j^\textup{II}, \min_{1 \leq i \leq j} (u_i^\textup{I} + b^\textup{II} + v_{j-i}^\textup{II})\right\},
		\blank\textup{and}\blank
		\textup{(b)}\blank v_j = \min_{i \leq j} (v_i^\text{I} + v_{j-i}^\text{II}).
	\end{equation}
\end{theorem}

\begin{proof}
	For all $\bs{q} \in \mb{U} \ua b$ and $j\in\mb{N}$, using (\ref{E:composability}) and (\ref{E:psiuv}), we have
	\[
	\begin{IEEEeqnarraybox}[][c]{rCl}
		\psi_j(\bs{q})
		&=& \psi_j^{(\bs{u}^\textup{II}, \bs{v}^\textup{II})}(\bs\psi^{(\bs{u}^\textup{I}, \bs{v}^\textup{I})}(\bs{q}-b^{\textup{II}}\bs\delta)+b^{\textup{II}}\bs\delta)\\
		&=& \min\left\{u_j^\text{II}, \min_{1 \leq i \leq j} (\psi_i^{(\bs{u}^\textup{I}, \bs{v}^\textup{I})}(\bs{q}-b^\text{II}\bs\delta) +b^\text{II}+v_{j-i}^\text{II})\right\}\\
		&=& \min\left\{u_j^\text{II}, \min_{1 \leq i \leq j} \left(\min\left\{u_i^\text{I}, \min_{1 \leq k \leq i}(q_k-b^\text{II}+v_{i-k}^\text{I})\right\} +b^\text{II}+v_{j-i}^\text{II}\right)\right\}\\
		&=& \min \left\{u_j^\textup{II}, \min_{1 \leq i \leq j} (u_i^\textup{I} + b^\textup{II} + v_{j-i}^\textup{II}), \min_{1 \leq i \leq j} \min_{1 \leq k \leq i} (q_k+v_{i-k}^\text{I}+ v_{j-i}^\textup{II})\right\}.
	\end{IEEEeqnarraybox}
	\]
	Using (\ref{E:uvComposing})-(a), and for the last term of the outer minimization, reversing the order of the two $\min$ operators, we have
	\[
		\psi_j(\bs{q}) = \min \left\{u_j, \min_{1 \leq k \leq j} 	\min_{k \leq i \leq j} (q_k+v_{i-k}^\text{I}+ v_{j-i}^\textup{II})\right\} = \min \left\{u_j, \min_{1 \leq k \leq j} 	\left(q_k + \min_{k \leq i \leq j} (v_{i-k}^\text{I}+ v_{j-i}^\textup{II})\right)\right\}.
	\]
	So, according to (\ref{E:uvComposing})-(b),
	\[
		\psi_j(\bs{q}) = \min \left\{u_j, \min_{1 \leq k \leq j} 	\left(q_k + \min_{l \leq j-k} (v_l^\text{I}+ v_{j-k-l}^\textup{II})\right)\right\} = \min \left\{u_j, \min_{1 \leq k \leq j} 	(q_k + v_{j-k})\right\}.
	\]
	Comparing this to (\ref{E:psiuv}), it is immediate that $\bs\psi(\bs{q}) = \bs\psi^{(\bs{u}, \bs{v})}(\bs{q})$, so $\bs\psi = \bs\psi^{(\bs{u}, \bs{v})}$.
\end{proof}

\end{document}